\documentclass{osa-article}
\journal{oe}
\usepackage{cases,bm,multirow}
\usepackage{setspace}
\usepackage{epsfig,amssymb,subfigure,version,graphicx,fancybox,mathrsfs,amscd}
\usepackage{url,xcolor}
\usepackage{mwe,tikz}
\usepackage[percent]{overpic}

\usepackage{amsthm}
\usepackage{hyperref}

\newtheorem{thm}{\bf Theorem}
\newtheorem{lem}{\bf Lemma}[section]
\newtheorem{rem}{ Remark}[section]
\graphicspath{{images/}}

\newenvironment{sumMy}[1][]{\textstyle\sum}

\articletype{Research Article}

\begin{document}

\title{Advanced denoising for X-ray ptychography}

\author{Huibin Chang,\authormark{1,2,10} Pablo Enfedaque,\authormark{2} Jie Zhang, \authormark{1} Juliane Reinhardt, \authormark{3,4}  Bjoern Enders, \authormark{5, 6}  Young-Sang Yu, \authormark{5}  David Shapiro, \authormark{5}  Christian G. Schroer, \authormark{7,8}  Tieyong Zeng \authormark{9}, Stefano Marchesini \authormark{2,11}
 }


\address{\authormark{1} School of Mathematical Sciences, Tianjin Normal University, Tianjin, China\\
\authormark{2}Computational Research Division, Lawrence Berkeley National Laboratory, Berkeley, CA, USA\\
\authormark{3} ARC Centre of Excellence for Advanced Molecular Imaging, Department of Chemistry and Physics,
La Trobe Institute for Molecular Sciences, La Trobe University, Bundoora, Australia\\
\authormark{4} The Australian Synchrotron, Clayton, Australia\\
\authormark{5} Advanced Light Source, Lawrence Berkeley National Laboratory, Berkeley, CA, USA\\
\authormark{6} Physics Department, University of California at Berkeley, Berkeley, CA, USA\\
\authormark{7} Deutsches Elektronen-Synchrotron DESY, D-22607 Hamburg, Germany\\
\authormark{8} Department Physik, Universit\"at Hamburg, Luruper Chaussee 149, 22761 Hamburg, Germany\\
\authormark{9} Department of Mathematics, The Chinese University of Hong Kong, Hong Kong\\
\authormark{10} \authormark{*}changhuibin@gmail.com \\
\authormark{11} \authormark{*}smarchesini@lbl.gov \\
}


\begin{abstract}
The success of ptychographic imaging experiments strongly  depends on achieving high signal-to-noise ratio. This is particularly important in nanoscale imaging experiments when diffraction signals are very weak and the experiments are accompanied by significant parasitic scattering (background), outliers or correlated noise sources. It is also critical when rare events such as cosmic rays, or bad frames caused by electronic glitches or shutter timing malfunction take place. 
 In this paper, we  propose a novel iterative algorithm with rigorous analysis that exploits the direct forward model for parasitic noise and {sample smoothness } to achieve a thorough characterization and removal of structured and random noise. We present a formal description of the proposed algorithm and prove its convergence  under mild conditions. Numerical experiments from simulations and real data (both soft and hard X-ray beamlines) demonstrate that the proposed algorithms produce better results when compared to state-of-the-art methods.
\end{abstract}

\section{Introduction}\label{intro}
Ptychography \cite{nellist1995resolution,rodenburg2004phase,maiden2009improved,jiang2018electron}
has become an increasingly popular imaging technique, and it is used nowadays in scientific fields as diverse as condensed matter physics \cite{shi2016soft}, cell biology \cite{giewekemeyer2010quantitative}, materials science \cite{shapiro2014chemical,yu2018three}, and electronics \cite{holler2017high,jiang2018electron}, among other areas. Compared to standard lens-based microscopy, the resolution in ptychography is not limited by {the size of the illumination probe}, but by the wavelength and number of photons used in an experiment. Unfortunately, multiple experimental challenges have to be tackled to achieve a high-quality reconstruction from a ptychographic experiment in practice. Complex experimental systems require nanometer stability while collecting large amounts of diffraction data, which can require several hours depending on the experimental setup. Data can be then contaminated by structured parasitic scattering \cite{wiedorn2017post,reinhardt2017beamstop,wang2017background}, detector read-out noise, Poisson noise derived from the photon counting, and different types of outliers. Additionally, the characterization of the illumination source is also commonly considered as a joint problem to the object reconstruction, known as Blind Ptychography (BP) \cite{maiden2009improved,thibault2009probe,hesse2015proximal,chang2018Blind}.


During the last decades, researchers have developed several schemes to solve the BP problem. Arguably, the most popular ones are extended Ptychographic Iterative Engine (ePIE) \cite{maiden2009improved},
Difference Map \cite{thibault2009probe,Elser2003}, Maximum Likelihood (ML) method \cite{thibault2012maximum}, Proximal  Splitting algorithm \cite{hesse2015proximal}, Relaxed Averaged Alternating Reflections (RAAR \cite{Luke2005}) based algorithms \cite{marchesini2016sharp}, and generalized  Alternating Direction Method of Multipliers (ADMM) \cite{glowinski1989augmented,Wu&Tai2010,boyd2011distributed}) based BP \cite{chang2018Blind}. Although some implementations of those methods present \textit{ad hoc} solutions for  structural noise removal, there is no formal analysis or in-depth characterization of such experimental problems, even though they can be critical to achieve robust high-resolution images, specially from weakly scattering or low contrast specimens. 

\paragraph{Experimental noise overview}

Direct reconstructions based on raw experimental data often contain visible artifacts. They commonly derive from outliers and structured and randomly distributed uncorrelated noise sources. Below there is a characterization of the main causes of such anomalies.


\noindent\hspace{0.7em}\textbf{\textit{Structured noise}}
\vspace*{-.0em}
\begin{itemize}
\vspace*{-.2em}
\item[\enspace\textit{(1)}]\textit{Parasitic scattering:} Derives from the scattering from any element along the beam path other than the sample and the optical elements desired harmonic order e.g.: slits, pinholes, lens imperfections, harmonic contamination, air, gas, etc. We refer to this as the \textit{background}. 
\vspace*{-.3em}
\item[\enspace\textit{(2)}]\textit{Saturation}: Occurs when the flux exceeds the detector capacity, the range of the Analog to Digital Converter (ADC) or the maximum photon counting rate. 
\vspace*{-.3em}
\item[\enspace\textit{(3)}]\textit{Dark noise}: Produced by the detector dark current, light from position encoders and interferometers, or light from the environment. It differs from the parasitic scattering as it occurs also when the x-ray beam is off and can be measured in advance by turning off the beam and (commonly referred to as \textit{dark frame}). The constant component can be subtracted from the measured signal before applying a reconstruction solver or alternatively it can be incorporated in the generalized \textit{background}.

\end{itemize}

\noindent\hspace{0.7em}\textbf{\textit{Random noise}}
\vspace*{-.0em}
\begin{itemize}
\vspace*{-.2em}
\item[\enspace\textit{(4)}]\textit{Photon-counting noise:} Caused by the quantization of the wavefront at the detector. 
\vspace*{-.3em}
\item[\enspace\textit{(5)}]\textit{Read-out and detector noise:} Derives from electronic interference when the ADC converts the charge distribution to a digital signal and the random component of the dark current. The interference and parallel ADC read-out electronics can produce ringing and correlation in the digital signal. 

\end{itemize}

\noindent\hspace{0.7em}\textbf{\textit{Outliers}}
\vspace*{-.0em}
\begin{itemize}
\vspace*{-.2em}
\item[\enspace\textit{(6)}]\textit{Bad frames:} Glitches on the functioning of mechanical components inside an instrument such as shutter timing can cause some measured frames to be corrupted.
\vspace*{-.3em}
\item[\enspace\textit{(7)}]\textit{Cosmic rays:} High energy cosmic particles that penetrate the atmosphere, the building and the instrument enclosure and hit the detector. Although they are rare, those types of signals normally present very high charge.
\vspace*{-.3em}
\item[\enspace\textit{(8)}]\textit{Bad pixels:} Result from fabrication errors may cause dead pixels on a detector, these can be typically measured in advanced and masked out during the reconstruction process.
\end{itemize}

\vspace*{5px}

\setstretch{1}

In this paper we focus on a solution to address \textit{(1)}, \textit{(6)} and \textit{(7)}, regarding structured noise and outliers, and also propose a technique to deal with a combination of different random noise distributions coming from \textit{(4)} and \textit{(5)}. The noise from \textit{(3)} and \textit{(8)} can be measured in advance, and it is commonly subtracted. Saturation noise \textit{(2)} can be equally identified and addressed by masking out the saturated signal areas and it is outside the focus of this paper. Random noise sources \textit{(4)} and \textit{(5)} have been considered in the literature \cite{godard2012noise,thibault2012maximum,odstrvcil2018iterative,chang2018Blind}, but there is no solution that can address the combination of two or more random noise distributions. In this paper, we propose a mixed-noise  term for ptychography to effectively deal with any combination of Poisson and Gaussian noise sources. 

Some of the previous pathologies can be identified inspecting the measured data. An illustrative example of parasitic noise~\textit{(1)} is reported in Fig.~\ref{fig0}. A single measured diffraction pattern (Fig. \ref{fig0}(a)) presents a distinct reverse ``L-Shape'' artifact, which consistently appears on all other diffraction measurements from the same dataset. 
If parasitic scattering is not considered, a standard ptychography reconstruction produces very noisy and low contrast images (Figs. \ref{fig0}(d) and \ref{fig0}(e)). The phase part is lost in some areas, concealed by the biased background of the recovered phase, and the resulted intensities (Fig. \ref{fig0}(b)) contain a slight alteration of the parasitic noise from the measured data. (A more detailed comparison between different methods, with and without structured noise correction, is provided in section \ref{Experimental results})

\begin{figure}[ht!]
\begin{center}
\subfigure[Diffraction pattern]{\includegraphics[width=.25\textwidth,trim={.8cm .6cm .6cm .8cm},clip]{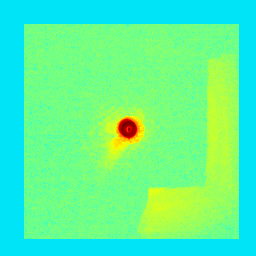}}
\subfigure[RAAR]{\includegraphics[width=.25\textwidth,trim={.8cm .6cm .6cm .8cm},clip]{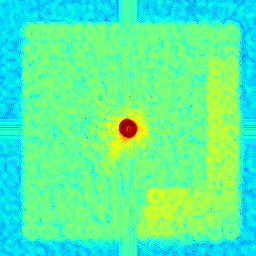}}
\subfigure[Proposed alg.]{\includegraphics[width=.25\textwidth,trim={.8cm .6cm .6cm .8cm},clip]{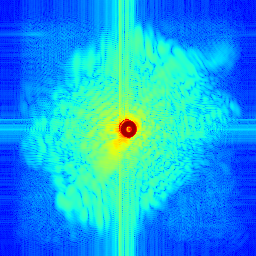}}
\subfigure{\includegraphics[width=.07\textwidth]{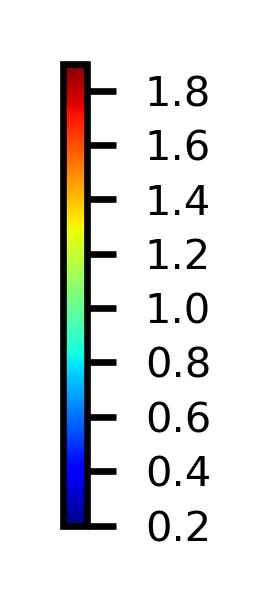}}
\addtocounter{subfigure}{-1}
\\
\subfigure{\includegraphics[width=.05\textwidth]{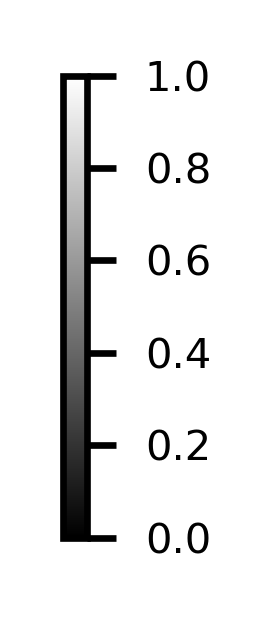}}
\addtocounter{subfigure}{-1}
\subfigure[RAAR (amp.)]{\includegraphics[width=.21\textwidth]{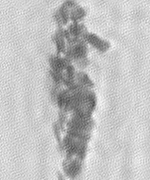}}
\subfigure[RAAR (phase)]{\includegraphics[width=.21\textwidth]    {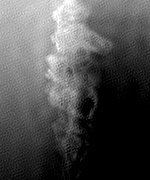}}
\subfigure[Proposed alg. (amp.)]{\includegraphics[width=.21\textwidth]{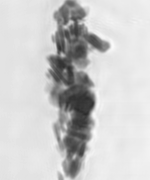}}
\subfigure[Proposed alg. (phase)  ]{\includegraphics[width=.21\textwidth]{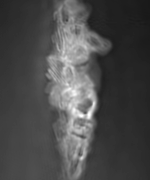}}
\subfigure{\includegraphics[width=.07\textwidth]{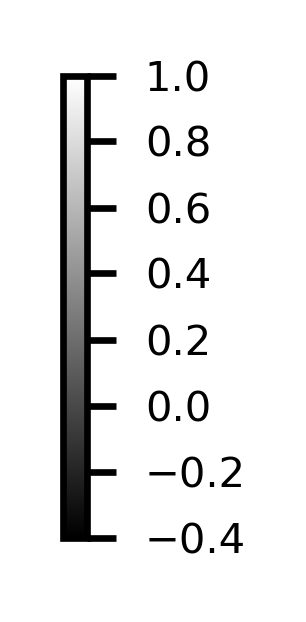}}
\addtocounter{subfigure}{-1}
\end{center}
\caption{Reconstruction results using RAAR (implemented in SHARP \cite{marchesini2016sharp}) and the proposed algorithm. The dataset corresponds to one of the 3D ptychography experiments from the results presented in~\cite{yu2018three}. (a) Measured intensities from a single diffraction pattern.  Recovered intensities (b), object amplitude (amp.) (d) and phase (e), respectively, using RAAR without structured noise correction.  Recovered intensities (c), object amplitude (f) and phase (g), respectively, using the proposed algorithm (alg.). {The figures in the first row are shown in scale $(\cdot)^{0.05}$ following the colorbar shown on the top right corner; the figures in the second row are shown in linear scale following the colorbars on the left corner (amplitude) and right corner (phase). }}\vspace{-.3em}
\label{fig0}
\end{figure}

\paragraph{Related work}

A beam stop scheme was employed in \cite{reinhardt2017beamstop} to reduce the parasitic noise using intensities from  two raster scans ({i.e. scanning on a Cartesian 2D grid}) with and/or without beamstop.
However, the artifacts caused by the structured noise are still noticeable in Fig. 2 of \cite{reinhardt2017beamstop}, and some fine features still appear undefined. A multiple-mode approach in \cite{enders2014ptychography} was proposed to reduce the background.
In \cite{wang2017background}, the data was  pre-processed by subtracting an adjustable scalar value multiplied by the \textit{dark frame}.
A method to remove the parasitic background automatically employed a gradient descent method with adaptive stepsizes was proposed in \cite{marchesini2013augmented} 
and implemented using a distributed multi-GPU implementation \cite{marchesini2016sharp} in a real-time streaming production code at a user facility \cite{daurer2017nanosurveyor}. 
 
 {
 Several researchers have also studied techniques to remove the random noise from photon-counting or read-out noises by employing a sparse regularization technique  to further improve the image quality for conventional ptychography, .e.g.  
Tikhonov regularization with nonlinear conjugate gradient method  \cite{thibault2012maximum}, total variation regularization with ADMM \cite{chang2016Total} and dictionary learning method with proximal algorithm \cite{chang2018denoising}.
Specially, for the more practical noises, like a mix of Poisson and Gaussian noises, several variational methods  have been proposed for linear inverse problems, e.g. generalized Anscombe transform \cite{murtagh1995image}, total variation based Shift-Poisson method \cite{chakrabarti2012image} and  weighted least squares method \cite{li2015reweighted}. However, such methods have not been applied to the conventional ptychography yet, which is essentially a nonlinear inverse problem.}  
Beyond algorithmic methods, there are also several contributions from beamline scientists attenuating the background problem via hardware solutions~\cite{kirby2013low,wiedorn2017post}. 


\paragraph{Motivations and contributions}

In a ptychography experimental setup it is impossible to isolate parasitic noise from the measured data. Furthermore, the data is often contaminated by outliers and an additional mix of random noise types. In this context, an automatic denoising algorithm with convergence guarantee to jointly remove both structured and random noise can be very valuable. 
In this paper, we propose ADMM Denoising for Phase retrieval (ADP), a new algorithm that addresses all the main sources of experimental noise in a blind X-ray ptychography experimental setting. The proposed algorithm achieves similar convergence speed as state-of-the-art blind ptychography methods and it is constructed on the framework presented in \cite{chang2018Blind}. Even if additional hardware or pre-processing techniques are in place, ADP can be used to further enhance the reconstruction quality of the retrieved object.
The main contributions of this work are listed below:
\begin{itemize}
\item[i.] 
 A new algorithm for automatic structured and random mixed noise removal with outliers correction. The algorithm is based on the forward physical model of the parasitic scattering noise ({additive frame-invariant background}) and on the maximum a posteriori (MAP) estimation combined with the shift-Poisson method \cite{chakrabarti2012image} for mixed types of noise. By assuming the piecewise smoothness of sample patches, We propose an additional framewise regularization term  to further enhance the denoising process and improve the image quality, instead of using a simple sparsity term of the sample~\cite{chang2016Total}. {To the best of our knowledge, it is the first time an iterative algorithm with rigorous analysis for background removal in ptychographic imaging is proposed.}

\item[ii.] By formulating the parasitic noise variable with positive constraint as a quadratic term,  fast ADMM algorithms are designed for the proposed model, with and without regularization. The convergence guarantee of the method is proved under very mild conditions. Computationally, each subproblem can be solved very efficiently, using pointwise operations or simple linear algebra solvers, requiring few arithmetic operations per iteration and exposing high parallelism. 

\item[iii.] Multiple numerical experiments on both simulated and experimental data from different light sources demonstrate the enhanced quality of the proposed algorithm. Reconstruction results using the proposed algorithm can be found in Fig.~\ref{fig0} and in the experimental results section.

\end{itemize}

\section{Proposed iterative algorithm}

\paragraph{Mathematical formula for standard ptychography}
In a standard ptychography experiment,  a localized coherent X-ray probe (or illumination) $\omega$  scans through an image (or sample) $u$, while the detector
collects a sequence of phaseless intensities  in the far field.
Throughout this paper, we consider the following discrete setting:

\vspace*{7px}
The variable \noindent$u\in\mathbb C^n$ ($\mathbb C^{n}$ is the complex Euclidean space) corresponds to a 2D sample with $\sqrt{n}\times\sqrt{n}$ pixels, and $\omega\in\mathbb C^{\bar m}$  is a localized 2D  probe   with  $\sqrt{\bar m}\times \sqrt{\bar m}$ pixels, both $u$ and $\omega$ written as  a vector by a lexicographical order. A stack of phaseless measurements ~$\tilde I_j\in \mathbb R_+^{\bar m}~\forall 0\leq j\leq J-1$ ($\mathbb R_+^{\bar m}$ is the Euclidean space with non-negative elements) is collected with $
\tilde I_j=|\mathcal F(\omega\circ \mathcal S_j u)|^2,
$
where $|\cdot|$ represents the element-wise absolute value of a vector, 
 $\circ$ denotes the element-wise multiplication, and $\mathcal F$ denotes the normalized discrete Fourier transformation. $\mathcal S_j\in \mathbb R^{\bar m\times n}$ is a binary  matrix that specifies a small window  with the index $j$ and  size  $\bar m$ over the entire sample $u$. The BP problem can then be expressed as follows:  
\begin{equation}\label{PtychoPR}
{\text{\qquad To find ~}\omega \in \mathbb C^{\bar m}\text{~and~} u\in \mathbb C^n,~~} \text{such that}~~|\mathcal A(\omega,u)|^2= \tilde I,
\end{equation}
where  bilinear operators  $\mathcal A:\mathbb C^{\bar m}\times \mathbb C^{n}\rightarrow \mathbb C^{m}$  and
 $\mathcal A_j:\mathbb C^{\bar m}\times \mathbb C^{n}\rightarrow \mathbb C^{\bar m}~\forall 0\leq j\leq J-1$,  are denoted as:
  $
 \mathcal A(\omega,u):=(\mathcal A_0^T (\omega,u), \mathcal A_1^T(\omega,u),\cdots, \mathcal A_{J-1}^T(\omega,u))^T,$
 $\mathcal A_j(\omega,u):=\mathcal F(\omega\circ \mathcal S_j u),$
 and $\tilde I:=({\tilde I}^T_0, {\tilde I}^T_1, \cdots, \tilde{I}^T_{J-1})^T\in \mathbb R^m_+$ ($m=J\times \bar m$).

\subsection{Proposed  model}
{
Experimentally, not all of the parasitic scattering  (background) is  the same on each frame (here, ``frame'' means  one image from the CCD detector). However, the frame-invariant component is one of the key sources of the artifacts in standard reconstruction results due to the redundancy of ptychography \cite{reinhardt2017beamstop,wang2017background}. Essentially,  it is an  additive components of signals, apart from those originating from the sample and the beam hitting it. } 
Hence, we start by assuming that the background  of each frame is approximately unchanged  \cite{marchesini2013augmented}, i.e., it is frame-invariant.
It is important to note that parasitic noise can also be interpreted as an independent structured noise source per frame. {The frame-invariant approximation is critical for proper modeling and the problem is actually not well defined if this approximation is not considered.  This is due to the fact that the independent noise per frame alternative has a trivial solution where the background is equal to the residual between the measured and reconstructed intensities from the iterative solution of the sample and probe. The frame-invariant approximation is, besides simple, very efficient at removing structural noise to produce high-contrast recovery images.} 

Formally, the true intensity  data $\tilde I$ is contaminated by non-negative parasitic structured noise $\hat\phi\in \mathbb R_+^{\bar m}$, such that the measured data  $\hat I=(\hat I_0^T,\hat I_1^T,\cdots,\hat I_{J-1}^T)^T\in \mathbb R_+^m$ with $\hat I_j\in \mathbb R_+^{\bar m}, \forall 0\leq j\leq J-1$ is recorded for each frame as follows:
\begin{equation}
\label{eqBkg}
0\leq \hat I_j=\tilde I_j+\hat\phi~\forall 0\leq j\leq J-1.
\end{equation}
In practice, the data is contaminated by noise as  
\begin{equation}
\label{eqBkgPoi}
\hat I_j=\mathrm{Noi}(\tilde I_j +\hat \phi)~\forall 0\leq j\leq J-1,
\end{equation}
where $\mathrm{Noi}$ denotes the degrading operator caused by different noise sources.
Without loss of generality, we  consider the Poisson and Gaussian mixed noise model with parasitic noise as:
\begin{equation}
\hat I_j(t)\stackrel{\mathrm{i.i.d}}{\sim}\mathrm{Poisson}(|\mathcal A_j(\omega, u)(t)|^2+\hat\phi(t))+\mathrm{n}_j(t),
\end{equation}
where $\mathrm{i.i.d.}$ is short for ``Independent and Identically Distributed'', $\mathrm{n}_j$ denotes the white Gaussian noise with zero mean and variance $\sigma>0.$ It is not difficult to verify that
$\mathrm{var}(\hat I_j(t)+\sigma^2)=\mathrm{var}(\hat I_j(t))=|\mathcal A_j(\omega, u)(t)|^2+\hat \phi+\sigma^2,$ and $\mathrm{mean}(\hat I_j(t)+\sigma^2)=|\mathcal A_j(\omega, u)(t)|^2+\hat\phi+\sigma^2$. Hence,  $I_j(t):=\hat I_j(t)+\sigma^2$ has the same variance and mean.  By further estimating $ I_j$  by a Poisson distribution  with the shift-Poisson technique \cite{chakrabarti2012image} and using KL divergence derived by maximum likelihood estimation of Poisson noise,
the following optimization problem is derived:
\begin{equation}
\min_{\omega,u,\phi} \tfrac12\sumMy\nolimits_j\left\langle \bm 1_{\bar m}, |\mathcal A_j(\omega, u)|^2+\phi-I_j\circ \log(|\mathcal A_j(\omega, u)|^2+\phi)\right\rangle+\mathbb I_{\mathscr E}(\phi),
\end{equation}
where  $\bm 1_{\bar m}$ is a vector of all ones, $\phi:=\hat\phi+\sigma^2$, $\langle\cdot,\cdot\rangle$ denotes the inner product   as $\langle v_1,v_2\rangle:=\sum\nolimits_{j=0}^{\bar m} v_1 (t)v_2(t)~\forall v_1, v_2\in \mathbb R^{\bar m}$, and positivity constraint set is defined as $\mathscr E:=\big\{\phi\in \mathbb R_+^{\bar m}:~\phi\geq 0\big\}.$ 
 Here, the indicator function $
\mathbb I_{\mathscr E}$ is defined as: 
$\mathbb I_{\mathscr E}(\phi)=0,$ if $\phi\in \mathscr E$; $\mathbb I_{\mathscr E}(\phi)=+\infty$, otherwise.

\begin{rem}
In order to jointly estimate the background  and illumination (double-blind), based  on Eq. \eqref{eqBkg}, a variant optimization model with least squares fitting can be established regardless of the specific noise mechanism as
\begin{equation}\label{eqGenModel}
\begin{split}
\min\limits_{\omega,u,\phi}& \tfrac{1}{2} \textstyle\sum\nolimits_{j=0}^{J-1} \big\||\mathcal A_j(\omega, u)|^2+\phi- I_j \big\|^2+\mathbb I_{\mathscr E}(\phi).
\end{split}
\end{equation}
However, when this metric is used to measure the distance between the collected and recovered intensities, the corresponding first-order iterative algorithm is significantly slow \cite{chang2018Blind}. This is because the residual slowly decreases as the iterations go, and the subproblem with respect to the auxiliary variable is related with quartic equations, which are solved with high computational cost \cite{chang2017variational}.
\end{rem}

In order to  deal with the misfit caused by experimental outliers, a weight vector $C:=(C_0^T,C_1^T,\cdots,C_{J-1}^T)^T\in\mathbb R^m$ is introduced to improve the fitting residual.
To further handle the noise or other artifacts, we consider the total variation regularization on the split patches of the sample, instead of the regularization of the entire sample in \cite{chang2016Total}.
As a result, we propose the following Kullback-Leibler (KL) divergence-based problem with framewise sparsity  and parasitic noise retrieval:
%

\begin{equation}\label{eqGenModel-II}
\begin{split}
&\min_{\omega, u,\phi}  \tfrac12 \sumMy\nolimits_j\Big\langle C_j, |\mathcal A_j(\omega, u)|^2+\phi+\varepsilon^2\bm 1_{\bar m} -(I_j+\varepsilon^2\bm 1_{\bar m})\circ\log(|\mathcal A_j(\omega, u)|^2+\phi+\varepsilon^2\bm 1_{\bar m})\Big\rangle\\
&\qquad+\lambda \textstyle\sum\nolimits_j \mathrm{TV}(\mathcal S_j u)+\mathbb I_{\mathscr E}(\phi),
\end{split}
\end{equation}
with a penalization parameter $\varepsilon>0,$
where $\mathrm{TV}$ denotes the total variation and $\lambda$ is a positive parameter to balance the regularization and data fitting terms.

%

%


\subsection{ADMM denoising for phase retrieval (ADP)}

{
Generally speaking, the ADMM is a variant of the augmented
Lagrangian method \cite{glowinski1989augmented,Wu&Tai2010}  that uses partial updates for the dual
variable, and each subproblem can be easily solved compared
with the original optimization problem. Due to its scalability
and flexibility, it has been a popular algorithm for large-scale
optimization problems arising in computer vision, statistics, 
machine learning, and other related areas, and see more details
in the review paper \cite{boyd2011distributed} and references therein.
 Since the proposed model is non-convex and non-smooth,  ADMM will be adopted to solve it, which allows for bigger stepsizes by avoiding directly calculating the gradient of the objective functional and therefore has fast convergence with low computation cost per iteration.
}

We formulate below the proposed ADP algorithm to solve Eq. \eqref{eqGenModel-II}.
First, we consider the problem without regularization by setting $\lambda=0$.
If directly following \cite{wen2012,chang2016phase}, the subproblem with respect to  the variable $\phi$ does not have a closed form solution:
\begin{equation}
\begin{split}
&\min\limits_{\omega,u,\phi} \tfrac12 \sumMy\nolimits_j\Big\langle C_j, |\mathcal A_j(\omega, u)|^2+\phi+\varepsilon^2\bm 1_{\bar m}-(I_j+\varepsilon^2\bm 1_{\bar m})\circ\log(|\mathcal A_j(\omega, u)|^2+\phi+\varepsilon^2\bm 1_{\bar m})\Big\rangle+\mathbb I_{\mathscr E}(\phi).
\end{split}
\end{equation}
%
A gradient descent or proximal type algorithm reported in \cite{chang2018Blind} to directly solve this problem presents very slow convergence. 
 In order to develop a more efficient algorithm  for each subproblem, we first rewrite  Eq. \eqref{eqGenModel-II} with $\lambda=0$ as
\begin{equation}\label{eqVarGenModel}
\begin{split}
&\min\limits_{\omega,u,\tilde\phi}  \tfrac12 \sumMy\nolimits_j\Big\langle C_j, |\mathcal A_j(\omega, u)|^2+\tilde\phi^2+\varepsilon^2\bm 1_{\bar m}-(I_j+\varepsilon^2\bm 1_{\bar m})\circ\log(|\mathcal A_j(\omega, u)|^2+\tilde\phi^2+\varepsilon^2\bm 1_{\bar m})\Big\rangle,
\end{split}
\end{equation}
where $\tilde\phi:=\sqrt{\phi}.$ Note that the constraint for $\phi$ is automatically removed in Eq. \eqref{eqVarGenModel}. Also, the variables $\mathcal A_j(\omega, u)$ and $\tilde\phi$ play now the same role in the objective function and adding auxiliary variables can derive a more efficient algorithm than directly solving the previous problem.  

Introducing the auxiliary variables $z_j=\mathcal A_j(\omega,u)$ and $\mu_j=\tilde\phi~~\forall 0\leq j\leq J-1,$
produces the following equivalent constraint optimization problem: 
\begin{equation}\label{eqVarGenModel-constr}
\begin{split}
\min\limits_{\omega, u,z,\mu,\tilde\phi}& \mathcal G_\varepsilon(z,\mu),~~\text{such that} ~z_j-\mathcal A_j(\omega,u)=0,~\mu_j-\tilde \phi=0~\forall 0\leq j\leq J-1,
\end{split}
\end{equation}
with
\begin{equation}
\begin{split}
&\mathcal G_\varepsilon(z,\mu):=
\tfrac12 \sumMy\nolimits_j\Big\langle C_j, |z_j|^2+\mu_j^2+\varepsilon^2\bm 1_{\bar m}-(I_j+\varepsilon^2\bm 1_{\bar m})\circ \log(|z_j|^2+\mu_j^2+\varepsilon^2\bm 1_{\bar m})\Big\rangle.
\end{split}
\end{equation}

The corresponding augmented Lagrangian for Eq. \eqref{eqVarGenModel-constr} reads:
\begin{equation}
\label{eqAL}
\begin{split}
&\mathcal L(\omega, u,z,\mu,\tilde\phi,\Lambda_1,\Lambda_2):=\mathcal G_\varepsilon(z,\mu)+r\sumMy_j\Re\langle \Lambda_{1,j},z_j-\mathcal A_j(\omega,u)\rangle\\
&+\tfrac{r}{2}\sumMy_j\|z_j-\mathcal A_j(\omega,u)\|^2+r\sumMy_j\langle \Lambda_{2,j},\mu_j-\tilde \phi\rangle+\tfrac{r}{2}\sumMy_j\|\mu_j-\tilde \phi\|^2,
\end{split}
\end{equation}
where $\Re$ denotes the real part of the complex-valued number.
The proposed ADP algorithm is formulated as follows:
\begin{equation}
\left\{
\begin{aligned}
\omega^{k+1}       &=\arg\min_{\omega}\mathcal L(\omega, u^k,z^k,\mu^k,\tilde\phi^k,\Lambda_1^k,\Lambda_2^k)+\tfrac{\alpha_1}{2}\|\omega-\omega^k\|^2;\\
u^{k+1}            &=\arg\min_{u}\mathcal L(\omega^{k+1}, u,z^k,\mu^k,\tilde\phi^k,\Lambda_1^k,\Lambda_2^k)+\tfrac{\alpha_2}{2}\|u-u^k\|^2;\\
(z^{k+1},\mu^{k+1})&=\arg\min_{z,\mu}\mathcal L(\omega^{k+1}, u^{k+1},z,\mu,\tilde\phi^k,\Lambda_1^k,\Lambda_2^k);\\
\tilde\phi^{k+1}   &=\arg\min_{\tilde\phi}\mathcal L(\omega^{k+1}, u^{k+1},z^{k+1},\mu^{k+1},\tilde\phi,\Lambda_1^k,\Lambda_2^k);\\
\Lambda_1^{k+1}&=\Lambda_1^k+z^{k+1}-\mathcal A(\omega^{k+1},u^{k+1}); \\
\Lambda_{2,j}^{k+1}&=\Lambda_{2,j}^k+\mu_{j}^{k+1}-\tilde\phi^{k+1}~\forall 0\leq j\leq J-1,
\end{aligned}
\right.
\end{equation}
with $\Lambda_1:=(\Lambda_{1,0}^T,\Lambda_{1,1}^T,\cdots,\Lambda_{1,J-1}^T)^T\in\mathbb C^m,$ and $\Lambda_2:=(\Lambda_{2,0}^T,\Lambda_{2,1}^T,\cdots,\Lambda_{2,J-1}^T)^T\in\mathbb R^m.$

The solution to each  subproblems above are reported in Algorithm 1, and further details can be found in the Appendix~\ref{apdx-0}.
When setting $\mu_j^0=\tilde\phi^0=0,$ and $\Lambda^0_2=0,$ following the above algorithm we have that $\mu_j^k=0~\forall k=1,2,\cdots,$ since $\mu^\star=0$ is a stationary point of the proposed model Eq. \eqref{eqGenModel-II}. Therefore, we use the following \textit{warm-start} scheme to initialize $\mu$: first solve a problem without structured noise correction by setting $\mu_j^0=0$, and after $J_0$ iterations reset $\mu_j^k$ using ${\mu_j^{k+1}:=\sqrt{\max\{0,\tfrac{1}{J}\textstyle\sum_j (I_j-|\mathcal A(\omega^{k+1},u^{k+1})|^2)\}}}.$ 
{The weight function $\{C_j\}$ can be simply fixed to be one. It can also be dynamically changed if the measured data contains outliers, such that the ADP algorithm is slightly modified with additional update of the weight function between Step 3 and Step 4,  and  please see more details for update scheme of this function in Remark \ref{rem1}.
}


\begin{center}
\vskip .15in
{
\hskip .2in
\hrule\vskip .1in
\centering  \textbf{Algorithm 1: ADP} \vskip .1in
\hrule\vskip .05in}
\begin{enumerate}
\item[\textbf{Step 0.}] Initialization: $u^0:=\bm 1_{n}$,
$\omega^0:=\tfrac{1}{J}\textstyle\sum_j \mathcal F^*\sqrt{I_j}$ ($\mathcal F^*$ denotes the normalized inverse discrete Fourier Transform),
$z_j^0:=\mathcal A_j(\omega^0, u^0), \mu^0=0$, $\Lambda_1=0, \Lambda_2=0$.
\item[\textbf{Step 1.}]
Refine the illumination by
\begin{equation}
\boxed{
\omega^{k+1}=\tfrac{\textstyle\sum_j\mathcal S_j (u^k)^*\circ \mathcal F^*(z^k_j+\Lambda^k_{1,j})+\alpha_1  \omega^k}{\textstyle\sum_j |\mathcal S_j u^k|^2+\alpha_1\bm 1_{\bar m}}.
}
\end{equation}

\item[\textbf{Step 2.}]
Refine the object by
\begin{equation}
\boxed{
u^{k+1}=\tfrac{\textstyle\sum_j\mathcal S_j^T((\omega^{k+1})^*\circ\mathcal F^*(z^k_j+\Lambda^k_{1,j}))+\alpha_2  u^k}{\textstyle\sum_j\mathcal S_j^T|\omega^{k+1}|^2+\alpha_2\bm 1_{n}}.
}
\end{equation}

\item[\textbf{Step 3.}] Update $z_j^{k+1}, \mu_j^{k+1}$ by
\begin{equation}
\boxed{
((z_j^{k+1})^T, (\mu_j^{k+1})^T)^T=((\rho_j^k)^T,(\rho_j^k)^T)^T\circ \mathrm{sign}(\bar X_j^k),
}
\end{equation}
with
\begin{equation}
(\bar X_j^k)^T:=(\mathcal A^T_j(\omega^{k+1},u^{k+1})-(\Lambda^k_{1,j})^T,(\tfrac{1}{J}\textstyle\sum\nolimits_j \mu^{k}_j)^T-(\Lambda^k_{2,j})^T)
\end{equation}
where  $\rho_j^k$ is computed as: 

\quad If $\varepsilon=0$:
\begin{equation}
{
\rho_j^k:=\tfrac{r|\bar X_j^k|_{*}+\sqrt{r^2|\bar X_j^k|_{*}^2+4(C_j+r\bm 1_{\bar m})\circ C_j\circ I_j}}{2(C_j+r\bm 1_{\bar m})};
}
\end{equation}
\quad Else if $\varepsilon>0$: it is iteratively determined by
\begin{equation}
\rho_{j, l+1}=\max\Big\{0, \Big((1-\tau r)\bm 1_{\bar m}-\tau C_j+\tau C_j\circ\tfrac{{I_j+\varepsilon^2\bm 1_{\bar m}}}{{\rho_{j,l}^2+\varepsilon^2\bm 1_{\bar m}}}\Big)\circ\rho_{j,l}+\tau r |\bar X_j^k|_*\Big\}
\end{equation}

 \qquad \enspace \enspace \enspace $l=0,1,\cdots,$ with $\rho_{j,0}:=\rho_j^{k-1}
 $
\item[\textbf{Step 4.}]
\begin{equation}\label{eqUpMul}
\boxed{
\begin{aligned}
&\Lambda^{k+1}_{1,j}=\Lambda^k_{1,j}+z^{k+1}_j-\mathcal A_j(\omega^{k+1},u^{k+1});\\
&\Lambda^{k+1}_{2,j}=\Lambda^k_{2,j}+\mu^{k+1}_j-\tfrac{1}{J}\textstyle\sum\nolimits_j \mu^{k+1}_j.
\end{aligned}
}
\end{equation}

\item[\textbf{Step 5.}]
If $k=J_0$: reinitialize the average of $\mu^{k+1}$ by \begin{equation}\boxed{\tfrac{1}{J}\sumMy\mu_j^{k+1}:=\sqrt{\max\{0,\tfrac{1}{J}\textstyle\sum_j (I_j-|\mathcal A(\omega^{k+1},u^{k+1})|^2)\}};  }\end{equation}
 Else if (\textit{satisfying some stopping condition}): output $u^{k+1}$ as the final reconstructed result
 Else: set $k:=k+1$, and goto Step 1
    \vskip .2in
\end{enumerate}
\hrule\vskip .05in
\end{center}
\vskip .1in

\paragraph{Convergence analysis}
We assess the convergence of ADP for the proposed model with $\lambda=0$ in the following theorem:
\begin{thm}
By denoting $Y^k:=(\omega^k,u^k,z^k,\mu^k,\tilde \phi^k,\Lambda_1^k,\Lambda_2^k),$ any limit point of $\{Y^k\}$ produced by ADP  is the stationary point of Eq. \eqref{eqGenModel-II} with $\lambda=0$ if the stepsize $r$ is sufficiently large and $\varepsilon>0$.
\label{thm1}
\end{thm}

The proof of the above theorem can be found in Appendix \ref{apdx-2}. 
The referred proof demonstrates that any limit point of the iterative sequence is the stationary point of the proposed model. In order to prove the convergence of the whole sequence we need more constraints for the object and illumination, and also a careful selection scheme of the proximal parameters $\alpha_1, \alpha_2$ \cite{chang2018Blind}.
{The introduction of $\varepsilon$ is also needed for the convergence analysis. 
Numerically, such penalization does not produce any obvious improvement on either convergence speed or reconstruction quality. Hence, for simplicity, we only show the performance of proposed algorithm with $\varepsilon=0$.
}
\subsection{ADP with regularization (ADPr)}
In this subsection we consider the proposed model~Eq. \eqref{eqGenModel-II} with $\lambda>0.$
Introducing the auxiliary variables $p_j=\nabla \mathcal S_j u,$ $z_j=\mathcal A_j(\omega,u)$ and $\mu_j=\tilde\phi~~\forall 0\leq j\leq J-1$
produces the following equivalent constraint optimization problem:
\begin{equation}\label{eqVarGenModel-constr-Reg}
\begin{split}
&\min\limits_{\omega, u,z,\mu,\tilde\phi} \lambda\textstyle\textstyle\sum\nolimits_j\big\||p_j|\big\|+\mathcal G_\varepsilon(z,\mu),~~\text{such that} ~z_j-\mathcal A_j(\omega,u)=0,~\mu_j-\tilde \phi=0,~p_j-\nabla \mathcal S_j u=0.
\end{split}
\end{equation}

The corresponding augmented Lagrangian reads
\begin{equation}
\begin{split}
&\mathcal L_{reg}(\omega, u,p,z,\mu,\tilde\phi,\Lambda_1,\Lambda_2,\Lambda_3)\\
:&=\mathcal G_\varepsilon(z,\mu)+\textstyle\sum\nolimits_j \Big(r\Re\langle \Lambda_{1,j},z_j-\mathcal A_j(\omega,u)\rangle+\tfrac{r}{2}\|z_j-\mathcal A_j(\omega,u)\|^2+r\langle \Lambda_{2,j},\mu_j-\tilde \phi\rangle\\
&+\tfrac{r}{2}\|\mu_j-\tilde \phi\|^2+\lambda\big\||p_j|\big\|+\beta\Re\langle p_j-\nabla \mathcal S_j u,\Lambda_{3,j}\rangle+\tfrac{\beta}{2}\|p_j-\nabla \mathcal S_j u\|^2\Big).
\end{split}
\end{equation}

We propose the following generalization of ADP to solve the problem above, referred to as ADP with regularization (ADPr):

\begin{equation}
\left\{
\begin{aligned}
\omega^{k+1}       &=\arg\min_{\omega}\mathcal L_{reg}(\omega, u^k,p^k,z^k,\mu^k,\tilde\phi^k,\Lambda_1^k,\Lambda_2^k,\Lambda_3^k)+\tfrac{\alpha_1}{2}\|\omega-\omega^k\|^2;\\
u^{k+1}            &=\arg\min_{u}\mathcal L_{reg}(\omega^{k+1}, u,p^k,z^k,\mu^k,\tilde\phi^k,\Lambda_1^k,\Lambda_2^k,\Lambda_3^k)+\tfrac{\alpha_2}{2}\|u-u^k\|^2;\\
p^{k+1}            &=\arg\min_{p} \mathcal L_{reg}(\omega^{k+1}, u^{k+1},p,z^k,\mu^k,\tilde\phi^k,\Lambda_1^k,\Lambda_2^k,\Lambda_3^k);\\
(z^{k+1},\mu^{k+1})&=\arg\min_{z,\mu}\mathcal L_{reg}(\omega^{k+1}, u^{k+1},p^{k+1},z,\mu,\tilde\phi^{k},\Lambda_1^k,\Lambda_2^k,\Lambda_3^k);\\
\tilde\phi^{k+1}   &=\arg\min_{\tilde\phi}\mathcal L_{reg}(\omega^{k+1}, u^{k+1},p^{k+1},z^{k+1},\mu^{k+1},\tilde\phi,\Lambda_1^k,\Lambda_2^k,\Lambda_3^k);\\
\Lambda_1^{k+1}&=\Lambda_1^k+z^{k+1}-\mathcal A(\omega^{k+1},u^{k+1}); \\
\Lambda_{2,j}^{k+1}&=\Lambda_{2,j}^k+\mu_{j}^{k+1}-\tilde\phi^{k+1}~\qquad\forall 0\leq j\leq J-1;\\
\Lambda_{3,j}^{k+1}&=\Lambda_{3,j}^{k+1}+p_j^{k+1}-\nabla\mathcal S_j u^{k+1}~\forall 0\leq j\leq J-1.
\end{aligned}
\right.
\end{equation}

Details of solvers for these subproblems are reported in Appendix \ref{apdx-0}. We only summarize  the overall algorithm in Algorithm 2 below.
\begin{center}
\vskip .15in
{
\hskip .2in
\hrule\vskip .1in
\centering  \textbf{Algorithm 2: ADP with regularization (ADPr)} \vskip .1in
\hrule\vskip .05in}
\begin{enumerate}
\item[\textbf{Step 0.}] Initialization: $u^0:=\bm 1_{n}$,
$\omega^0:=\tfrac{1}{J}\textstyle\sum_j \mathcal F^*\sqrt{I_j},$
$z_j^0:=\mathcal A_j(\omega^0, u^0)$, $\Lambda_1=0, \Lambda_2=0, \Lambda_3=0, p^0=0,$ and ${\mu_j^0:=0}$

\item[\textbf{Step 1.}]
Refine the illumination as Step 1 of Algorithm 1

\item[\textbf{Step 2.}]
Refine the object by solving the following equation using the conjugate gradient method
\begin{equation}
\boxed{
\begin{aligned}
&\qquad\qquad\qquad\textstyle\sum\nolimits_j\big( \mathrm{diag}(r\mathcal S_j^T |\omega^{k+1}|^2+\alpha_2\bm 1_n)-\beta \mathcal S_j^T \Delta\mathcal S_j\big) u^{k+1}\\
&=\textstyle\sum\nolimits_j\left( r\mathcal S_j^T((\omega^{k+1})^*\circ \mathcal F^*(z^k_j+\Lambda^k_{1,j}))-\beta\mathcal S_j^T\mathrm{div}(p^k_j+\Lambda^k_{3,j})\right)+\alpha_2 u^k,
\end{aligned}
}
\end{equation}
where $\mathrm{diag(\cdot)}$ gives a square diagonal matrix with the elements of a vector on the main diagonal, and $\mathrm{div}(\cdot)$ denotes the discrete divergence operator. 
\item[\textbf{Step 3.}] Refine the gradient of the object by
\begin{equation}
\boxed{
p^{k+1}_j=\max\{ 0, |\nabla\mathcal S_j u^{k+1}-\Lambda^k_{3,j}|-\tfrac{\lambda}{\beta}\bm 1_{\bar m} \}\circ \tfrac{\nabla\mathcal S_j u^{k+1}-\Lambda^{k}_{3,j}}{|\nabla\mathcal S_j u^{k+1}-\Lambda^k_{3,j}|}.
}
\end{equation}

\item[\textbf{Step 4.}] Update $z_j^{k+1}, \mu_j^{k+1}$ as Step 3 of Algorithm 1
\item[\textbf{Step 5.}] Update $\Lambda_1^{k+1}, \Lambda_2^{k+1}$ as Step 4 Algorithm 1, and update $\Lambda_3^k$
by
\begin{equation}\label{eqUpMul-Reg}
\boxed{
\begin{aligned}
&\Lambda^{k+1}_{3,j}=\Lambda^k_{3,j}+p^{k+1}_j-\nabla \mathcal S_j u^{k+1}.
\end{aligned}
}
\end{equation}

\item[\textbf{Step 6.}] If $k=J_0$: reinitialize the average of $\{\mu_j^{k+1}\}$ as Step 5 of Algorithm 1;


 Else if (\textit{satisfying some stopping condition}): output $u^{k+1}$ as the final reconstructed result
 
 Else: set $k:=k+1$, and goto Step 1.
    \vskip .2in
\end{enumerate}

\hrule\vskip .05in
\end{center}
\vskip .1in

In a similar manner as with ADP, one can derive the convergence of ADPr. We provide the following theorem omitting the proof details in this case:
\begin{thm}\label{thm2}
Any limit point of the iterative sequence generated by ADPr  is the stationary point of Eq. \eqref{eqGenModel-II} if the stepsizes $r,\beta$ are sufficiently large and $\varepsilon>0$.
\end{thm}
%
%
\begin{rem}
\label{rem1}
In order to handle outliers, we propose the following adaptive weights $\{C_j\}$. They are designed to produce bigger values when the residual between the measured data and recovered intensities are smaller.
%
The weight function $C_j$ is given based on the inverse of the residual as follows:
\begin{equation}
C_j\leftarrow\tfrac{2}{\gamma}
{\Big|\tfrac{1}{\sqrt{|z^{k+1}_j|^2+|\mu^{k+1}_j|^2+\varepsilon^2\bm 1_{\bar m}}-\sqrt{I_j+\varepsilon^2\bm 1_{\bar m}}}\Big|^{2-\gamma}},
\end{equation}
with $\gamma\leq 2.$
The weight function can be updated between Step 3 and Step 4 in ADP algorithm, or Step 4 and Step 5 in ADPr algorithm.

\end{rem}

\begin{rem}
Raster grid scanning can cause visible periodical artifacts, and therefore a constraint of the lens can be enforced as \cite{marchesini2016sharp}.
 Furthermore, a detector mask can be used as well.  

\end{rem}

\section{Experimental results}\label{Experimental results}

The experimental results of this section are generated using both simulated and experimental ptychography data. In most experiments, the performance of the proposed algorithm is compared with that of SHARP (using RAAR)~\cite{marchesini2016sharp}, a ptychography software solution that implements state-of-the-art blind ptychography algorithms and background retrieval techniques~\cite{marchesini2013augmented} (further details can be found in Appendix A). In some test we will evaluate SHARP without and with background retrieval, the later will be referred to as SHARP-B.
All experimental results in this section employ raster scans.  
 Regarding the parameter initialization of the proposed method, we set $J_0=5$ (after $J_0$ iterations reinitialize the background) for both ADP and ADPr. We use 5 iterations for Step 2  of ADPr. 

 We introduce two different criteria to evaluate performance. For simulation data,  the ground truth can be used to compute the signal-to-noise ratio (SNR) of $u^k$ as:
 \begin{equation}
 \mathrm{SNR}(u^k, u_g)=-10\log_{10}{\sum\limits_{t=0}^{n-1}|\zeta^* u^k(t+T^*)-u_g(t)|^2}/{\|\zeta^* u^k\|^2},
 \end{equation}
where $u_g$ corresponds to the ground truth image. The error is computed up to the translation $T^*$, and the phase shift and scaling factor $\zeta^*$ are determined by:
 \begin{equation}
 (\zeta^*,T^*):=\arg\min\limits_{\zeta\in \mathbb C, T\in \mathbb Z}\sum_{t}|\zeta u^k(t+T)-u_g(t)|^2.
 \end{equation}
For experiment data, the R-factor
for ADP and ADPr is used, defined
as:
\begin{equation}
\text{R-factor}^k:= \tfrac{\sum\nolimits_j\big\|\sqrt{|\mathcal A_j (\omega^k, u^k)|^2+(\tilde\phi^k)^2}-\sqrt{I_j} \big\|_1}{\|\sqrt{I}\|_1}.
\end{equation}

The R-factor of  SHARP is defined by setting $\tilde\phi=0$, whereas the R-factor of SHARP-B is defined in the same way as with ADP and ADPr.



\subsection{Synthetic data}

Given a true illumination and sample $\omega^\star$ and $u^\star$, respectively, and $\phi^\star$ as the parasitic noise component, the simulated intensities are generated as follows:
\begin{equation}
\label{eq:synt}
\hat I_j(t)\stackrel{\mathrm{i.i.d}}{\sim}\mathrm{Poisson}\left ( |\mathcal A_j(\omega^\star, u^\star)|^2(t)+\phi^\star(t) \right )+\mathrm{n}_j(t),~\forall~0\leq t\leq \bar m-1,
\end{equation}
where $\mathrm{n}_j$ denotes the white Gaussian noise {(the variance is set to $1.0\times 10^{-6}\times\tfrac{1}{m}\sum_t I_{poi}(t)$, with $I_{poi}(t):=\mathrm{Poisson}\left ( |\mathcal A_j(\omega^\star, u^\star)|^2(t)+\phi^\star(t) \right )$ )}. We also consider a further corruption of the intensities by outliers. They are simulated as patches of bad measurements appearing in 10\% of the frames, with size $\sim$ 50\% of the frame size. The corrupted intensity values are set to 0.
We use raster grid scan with stepsizes of 16 pixels, so an additional constraint of the illumination (a.k.a. support of the lens or illumination Fourier mask) is used to prevent potential periodical artifacts. In the experiments below, the frame size employed is $64\times 64$ pixels and all algorithms stop after 1000 iterations. 

Fig.~\ref{fig1-1} presents reconstruction results employing two different simulation datasets. The first experiment (second and fourth rows) uses data contaminated only by parasitic noise ($\phi^\star$). The second experiment (third and fifth rows) reconstructs data with a mix of white Gaussian  and Poisson noises as in Eq. \eqref{eq:synt} and  the outliers described above as well. The ground truth images for the sample, illumination and background are depicted in Figs.~\ref{fig1-1}(c), \ref{fig1-1}(k), Figs.~\ref{fig1-1}(a)-\ref{fig1-1}(b),  and Fig. \ref{fig1-2}(a), respectively. For each of the previous reconstruction experiments, we report the retrieved background in Fig.~\ref{fig1-2} for SHARP-B, ADP and ADPr. When no mix of noises and outliers are considered, ADP produces much higher quality reconstructions (Figs. \ref{fig1-1}(f) and \ref{fig1-1}(n)) than SHARP and SHARP-B (Figs. \ref{fig1-1}(d)-\ref{fig1-1}(e) and Figs. \ref{fig1-1}(l)-\ref{fig1-1}(m)). The retrieved background, for both SHARP-B and ADP (Figs.~\ref{fig1-2}(b)-\ref{fig1-2}(c)), is consistent with those results. When additional outliers are introduced,  ADP  still generates better result with a dramatic reduction in outliers-related artifacts (Figs. \ref{fig1-1}(i) and \ref{fig1-1}(q)). In this case, ADPr (Figs. \ref{fig1-1}(j) and \ref{fig1-1}(r)) helps further removing the noise from the baseline ADP reconstruction. For the second experiment, the proposed algorithms also retrieve the background (Figs. \ref{fig1-2}(e)-\ref{fig1-2}(f)) more precisely than SHARP-B (Fig.~\ref{fig1-2}(d)). {The SNRs of recovery results are reported:  $\mathrm{SNR}=13.2, 15.7, 76.7$ by SHARP, SHARP-B, and ADP  for the  case of single parasitic noise;  $\mathrm{SNR}=9.8, 10.2, 18.5, 27.9$ by SHARP, SHARP-B, ADP and ADPr  for the case with additional mixture noises and outliers. Again, one can readily  see the advantages of proposed ADP and ADPr.}

\begin{figure}
\vskip -.2in
\begin{center}
\subfigure{\includegraphics[width=.07\textwidth]{images/_colorbar}}
\addtocounter{subfigure}{-1}
\subfigure[True illum.(amp.)]{\includegraphics[width=.18\textwidth]{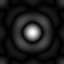}}\qquad
\subfigure[True illum.(phase) ]{\includegraphics[width=.18\textwidth]{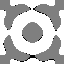}}\qquad
\subfigure{\includegraphics[width=.07\textwidth]{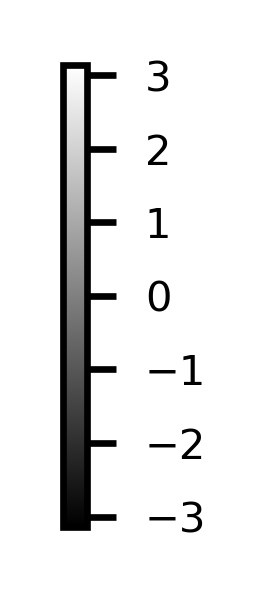}}
\addtocounter{subfigure}{-1}\\
\subfigure{\includegraphics[width=.05\textwidth]{images/_colorbar}}
\addtocounter{subfigure}{-1}
\subfigure[True image $(u^\star)$(amp.)]{\includegraphics[width=.23\textwidth]{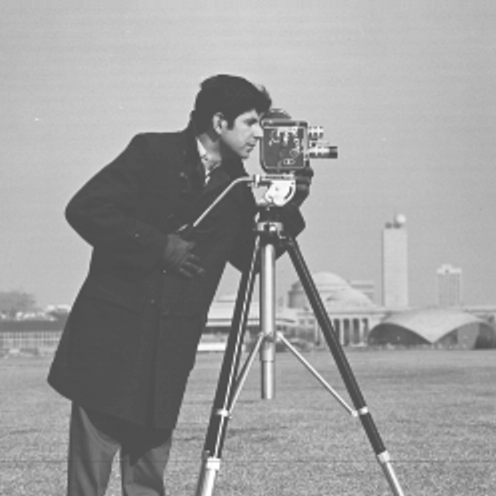}}
\subfigure[SHARP \textit{(amp.)}]    {\includegraphics[width=.23\textwidth]{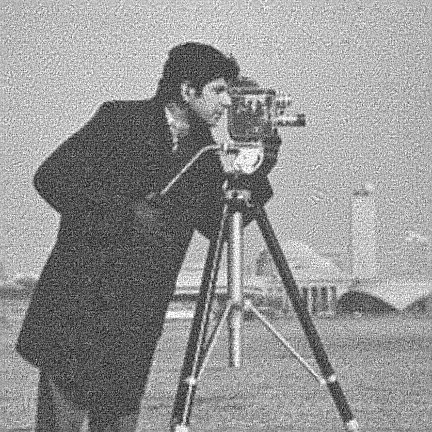}}
\subfigure[SHARP-B \textit{(amp.)}]    {\includegraphics[width=.23\textwidth]{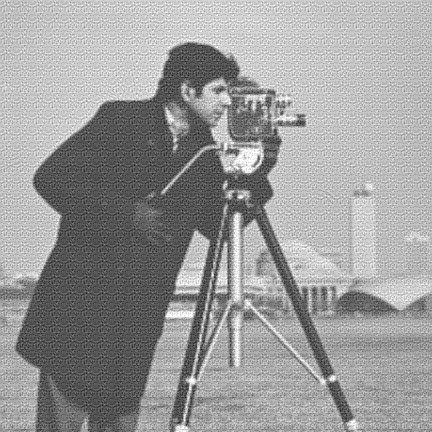}}
\subfigure[ADP \textit{(amp.)}]   {\includegraphics[width=.23\textwidth]{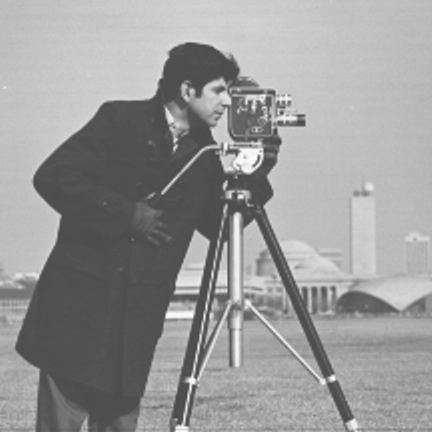}}~
\\
\hskip .8cm
\subfigure[SHARP \textit{(amp.)}]{\includegraphics[width=.23\textwidth]{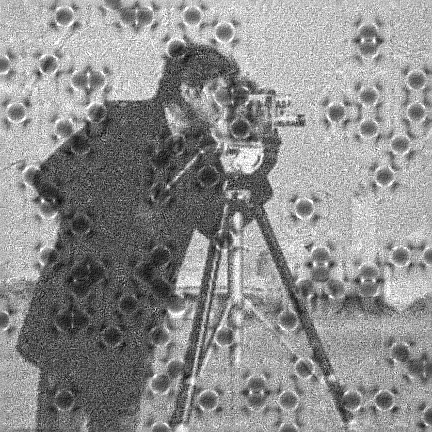}}
\subfigure[SHARP-B \textit{(amp.)}]{\includegraphics[width=.23\textwidth]{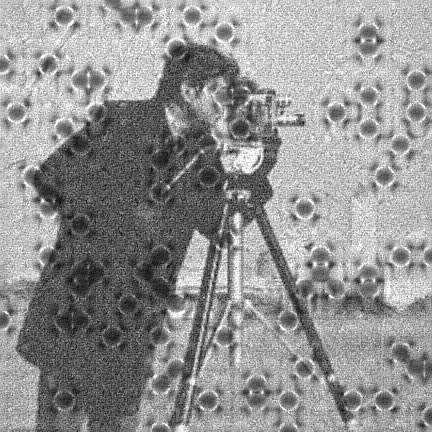}}
\subfigure[ADP \textit{(amp.)}]{\includegraphics[width=.23\textwidth]{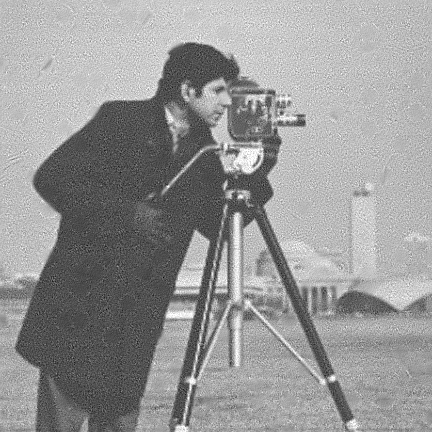}}
\subfigure[ADPr \textit{(amp.)}]{\includegraphics[width=.23\textwidth]{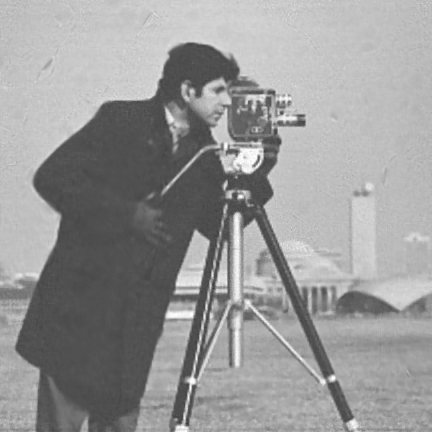}}
\\
\subfigure{\includegraphics[width=.05\textwidth]{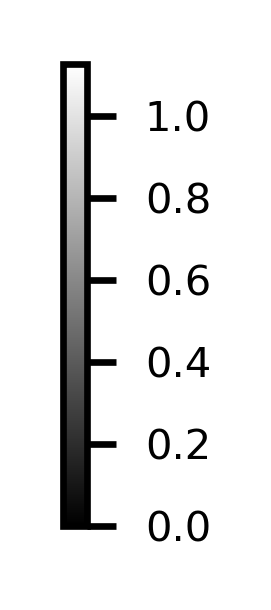}}
\addtocounter{subfigure}{-1}
\subfigure[True image $(u^\star)$(phase)]{\includegraphics[width=.23\textwidth]{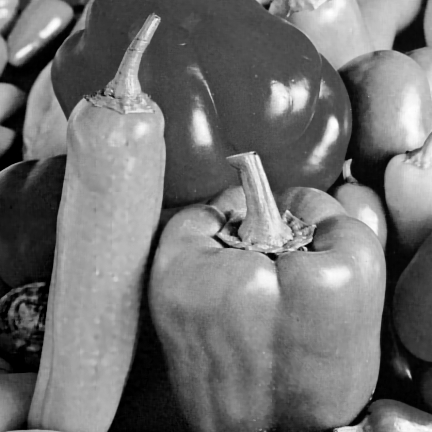}}
\subfigure[SHARP \textit{(phase)}]    {\includegraphics[width=.23\textwidth]{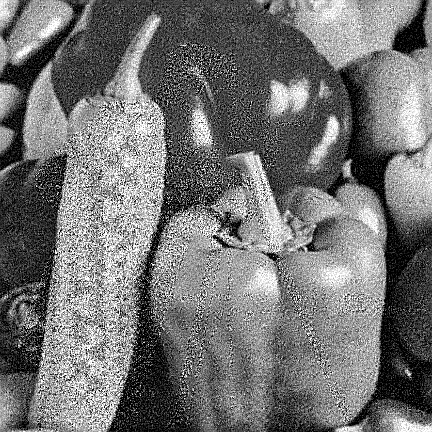}}
\subfigure[SHARP-B \textit{(phase)}]    {\includegraphics[width=.23\textwidth]{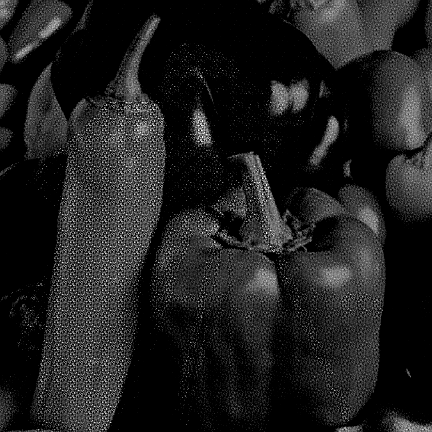}}
\subfigure[ADP \textit{(phase)}]   {\includegraphics[width=.23\textwidth]{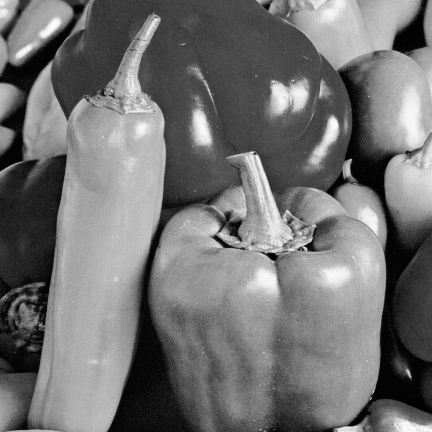}}~
\\
\hskip .8cm
\subfigure[SHARP \textit{(phase)}]{\includegraphics[width=.23\textwidth]{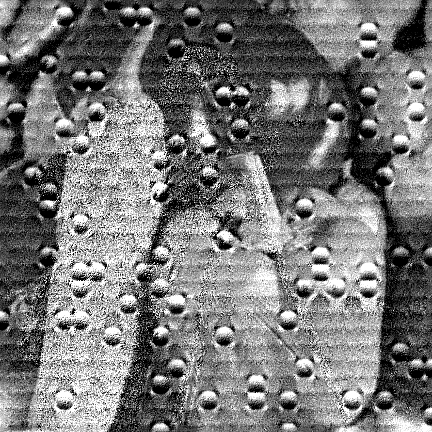}}
\subfigure[SHARP-B \textit{(phase)}]{\includegraphics[width=.23\textwidth]{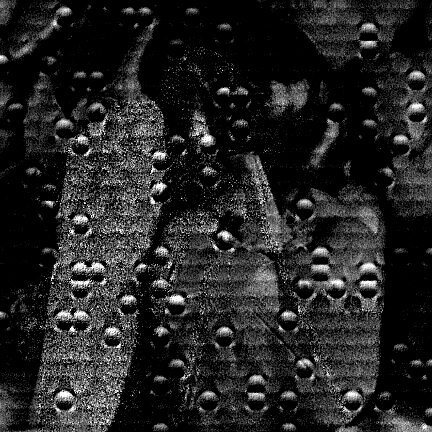}}
\subfigure[ADP \textit{(phase)}]{\includegraphics[width=.23\textwidth]{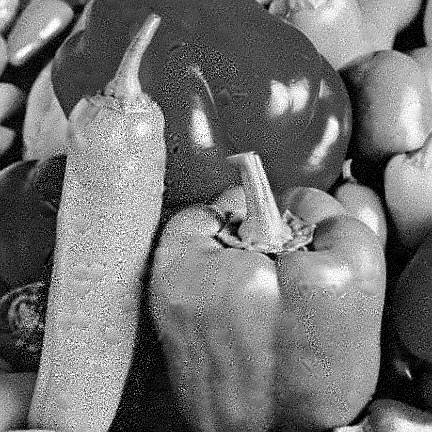}}
\subfigure[ADPr \textit{(phase)}]{\includegraphics[width=.23\textwidth]{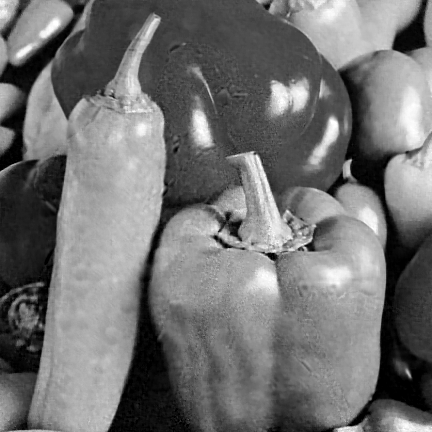}}
\\
\end{center}
\caption{Simulation experiment. First row: (a)-(b)  Amplitude and phase for true illumination (illu., $64\times 64$ pixels),  with Full Width at Half Max (FWHM) = 7 pixels. Second row: Amplitude of true image (c) with $496\times 496$ pixels, and amplitude of recovered images with data contaminated by only parasitic noise, using SHARP (d), SHARP-B (e) and ADP (f). Third row: Amplitude of recovered images with data contaminated as in~Eq. \eqref{eq:synt} and also by additional outliers, using SHARP (g), SHARP-B (h), ADP (i), and ADPr (j). The fourth, and fifth rows show the corresponding phase parts of the images in the second and third rows.}
\label{fig1-1}
\end{figure}

\begin{figure}
\begin{center}
\subfigure{\includegraphics[width=.05\textwidth]{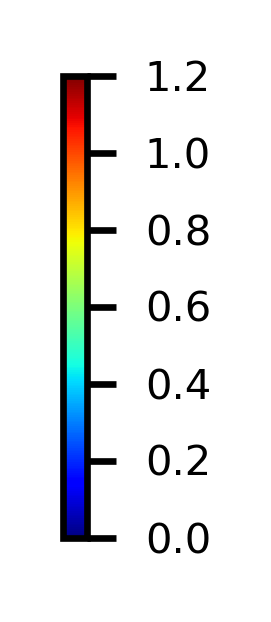}}
\addtocounter{subfigure}{-1}
\subfigure[True backg.]{\includegraphics[width=.16\textwidth]{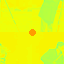}}
\subfigure[SHARP-B]{\includegraphics[width=.16\textwidth]{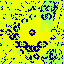}}
\subfigure[ADP]{\includegraphics[width=.16\textwidth]{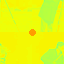}} \\
\hskip .8cm
\subfigure[SHARP-B]{\includegraphics[width=.16\textwidth]{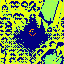}}
\subfigure[ADP]{\includegraphics[width=.16\textwidth]{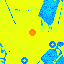}}
\subfigure[ADPr]{\includegraphics[width=.16\textwidth]{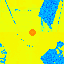}}\\
\end{center}
\caption{Retrieved backgrounds (backg.) for the simulation experiment of Fig.~\ref{fig1-1}. First row: True background (parasitic noise) (a) and results retrieved in a simulation only with parasitic noise for SHARP-B (b) and ADP (c). Second row: Background retrieved in a simulation as in~Eq. \eqref{eq:synt} with additional outliers for SHARP-B (d), ADP (e) and ADPr (f).{The figures are shown in scale $(\cdot)^{0.05}$ with the colorbar shown on the top left corner}}
\label{fig1-2}
\end{figure}

\subsection{Experimental data}

\paragraph{Soft X-ray dataset from the ALS, Lawrence Berkeley Lab}
The data used in the following experiment was published in \cite{yu2018three}. It corresponds to a 3D ptychography imaging experiment of battery cells at 708 eV, obtained from different projection angles.

The recovered amplitude and phase  are reported in Fig. \ref{fig4} and zoom-in view in Fig. \ref{fig4-zoom} when using SHARP, SHARP-B, ADP and ADPr. The recovered intensities of a single frame and the recovered background for the same experiment are shown in  Fig. \ref{fig4-1}.
When no background retrieval is used, the recovery results from SHARP (Figs. \ref{fig4}(a) and \ref{fig4}(e), and zoom-in view in Figs. \ref{fig4-zoom}(a) and \ref{fig4-zoom}(e) ) are significantly noisy, with some areas presenting vague or inappreciable features. SHARP-B was specially designed for structured noise from ALS soft X-ray sources and the results in Figs.~\ref{fig4}(b) and \ref{fig4}(f), and zoom-in view in Figs. \ref{fig4-zoom}(b) and \ref{fig4-zoom}(f),  show a considerable improvement with respect to SHARP baseline. Visually, ADP produces even sharper results and cleaner background than SHARP-B, specially when employing regularization (Figs.~\ref{fig4}-\ref{fig4-zoom} last two columns). This can also be appreciated by inspecting 
the recovered intensities and backgrounds from Fig~\ref{fig4-1}. Artifacts due to some bad frames and other source of noise can be clearly appreciated in Figs. \ref{fig4}(a)-\ref{fig4}(b), \ref{fig4}(e)-\ref{fig4}(f), and Figs. \ref{fig4-zoom}(a)-\ref{fig4-zoom}(b), \ref{fig4-zoom}(e)-\ref{fig4-zoom}(f) in the experiments with SHARP and SHARP-B. When using the proposed algorithms, such artifacts are greatly attenuated (in Figs. \ref{fig4}(c)-\ref{fig4}(d), and Figs. \ref{fig4-zoom}(c)-\ref{fig4-zoom}(d)). 

The R-factors of the previous experiment are 0.2399, 0.1112, 0.0957, and 0.0969, when using SHARP, SHARP-B, ADP and ADPr, respectively. This demonstrates how the proposed algorithm achieves smaller residuals thus producing higher accuracy results.
We also provide cutline values results from the previous experiment in Fig.~\ref{figCutline}, selecting a line from the center of the reconstructed image. It can be seen from the figure how the proposed algorithm generates a higher contrast reconstruction than SHARP-B, specially in the retrieved phase.
Convergence results are presented in Fig. \ref{figRf}, showing the R-factor as the iterations go. Those results illustrate how the R-factors of the proposed algorithm steadily decrease, providing additional evidence on the stability of the method. 

\begin{figure*}
\begin{center}
\subfigure{\includegraphics[width=.045\textwidth]{images/_colorbar}}
\hskip -.1in
\addtocounter{subfigure}{-1}
\subfigure[SHARP]{\includegraphics[width=.23\textwidth]  {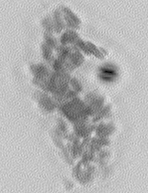}}
\subfigure[SHARP-B]{\includegraphics[width=.23\textwidth]{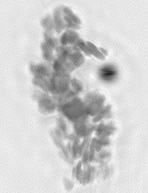}}
\subfigure[ADP]{\includegraphics[width=.23\textwidth]   {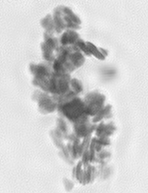}}
\subfigure[ADPr]{\includegraphics[width=.23\textwidth]  {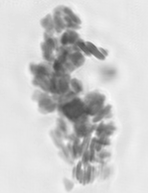}}
\\
\subfigure{\includegraphics[width=.045\textwidth]{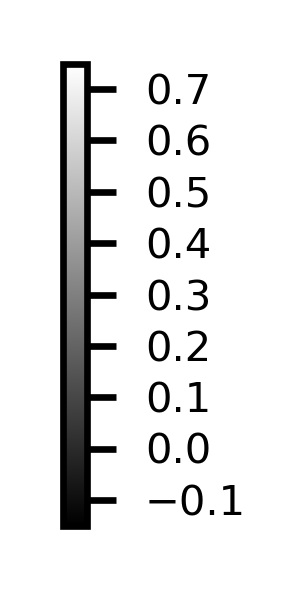}}
\addtocounter{subfigure}{-1}
\hskip -.1in
\subfigure[SHARP]{\includegraphics[width=.23\textwidth]  {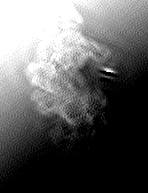}}
\subfigure[SHARP-B]{\includegraphics[width=.23\textwidth]{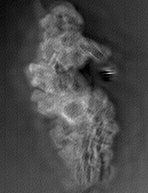}}
\subfigure[ADP]{\includegraphics[width=.23\textwidth]   {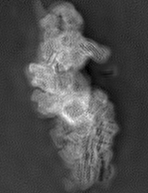}}
\subfigure[ADPr]{\includegraphics[width=.23\textwidth]  {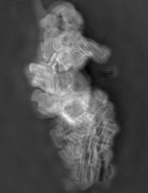}}
\end{center}
\caption{Soft X-ray experimental results from the data presented in~\cite{yu2018three}. First row: reconstructed amplitude using SHARP (a), SHARP-B (b), ADP (c) and ADPr (d). Second row:
 reconstructed phase using SHARP (e), SHARP-B (f), ADP (g) and ADPr (h).}
\label{fig4}
\end{figure*}

\begin{figure*}
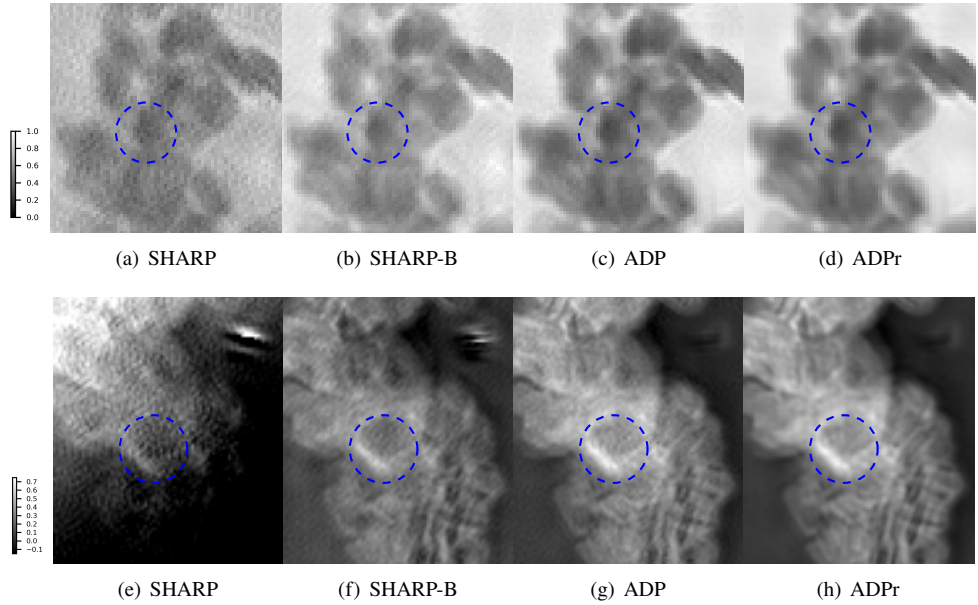

\begin{center}
\subfigure{\includegraphics[width=.05\textwidth]{images/_colorbar}}
\hskip -.1in
\addtocounter{subfigure}{-1}
          \subfigure[SHARP]
          {
     \begin{overpic}[width=.23\textwidth,trim={1cm 3.4cm 1.8cm 1cm},clip]{images/eg4/eg4-1501300570431M_obj}
     \put(25,30){
     \tikz\draw[blue,thick,dashed] (0,0) circle (.4);
     }
     \end{overpic}
     }
          \hskip -.1in
          \subfigure[SHARP-B]
          {
     \begin{overpic}[width=.23\textwidth,trim={1cm 3.4cm 1.8cm 1cm},clip]{images/eg4/eg4-150130057043Sharp_obj}
     \put(25,30){
     \tikz\draw[blue,thick,dashed] (0,0) circle (.4);
     }
     \end{overpic}
     }
          \hskip -.1in
          \subfigure[ADP]
          {
     \begin{overpic}[width=.23\textwidth,trim={1cm 3.4cm 1.8cm 1cm},clip]{images/eg4/eg4-150130057043100_obj}
     \put(25,30){
     \tikz\draw[blue,thick,dashed] (0,0) circle (.4);
     }
     \end{overpic}
     }
          \hskip -.1in
          \subfigure[ADPr]
          {
     \begin{overpic}[width=.23\textwidth,trim={1cm 3.4cm 1.8cm 1cm},clip]{images/eg4/eg4-150130057043400_obj}
     \put(25,30){
     \tikz\draw[blue,thick,dashed] (0,0) circle (.4);
     }
     \end{overpic}
     }
          \\
\subfigure{\includegraphics[width=.05\textwidth]{images/eg4/eg4-150130057043_p_colorbar}}
\hskip -.1in
\addtocounter{subfigure}{-1}
            \subfigure[SHARP]
          {
     \begin{overpic}[width=.23\textwidth,trim={1.2cm 1.2cm .9cm 2.0cm},clip]{images/eg4/eg4-1501300570431M_obj_p}
     \put(22,30){
     \tikz\draw[blue,thick,dashed] (0,0) circle (.45);
     }
     \end{overpic}
     }
     \hskip -.1in
          \subfigure[SHARP-B]
          {
     \begin{overpic}[width=.23\textwidth,trim={1.2cm 1.2cm .9cm 2.0cm},clip]{images/eg4/eg4-150130057043Sharp_obj_p}
     \put(22,30){
     \tikz\draw[blue,thick,dashed] (0,0) circle (.45);
     }
     \end{overpic}
     }
          \hskip -.1in
          \subfigure[ADP]
          {
     \begin{overpic}[width=.23\textwidth,trim={1.2cm 1.2cm .9cm 2.0cm},clip]{images/eg4/eg4-150130057043100_obj_p}
     \put(22,30){
     \tikz\draw[blue,thick,dashed] (0,0) circle (.45);
     }
     \end{overpic}
     }
          \hskip -.1in
          \subfigure[ADPr]
          {
     \begin{overpic}[width=.23\textwidth,trim={1.2cm 1.2cm .9cm 2.0cm},clip]{images/eg4/eg4-150130057043400_obj_p}
     \put(22,30)
     {
     \tikz\draw[blue,thick,dashed] (0,0) circle (.45);
     }
     \end{overpic}
     }
\end{center}
\caption{Zoom-in view of parts of Fig. \ref{fig4}. First row: reconstructed amplitude using SHARP (a), SHARP-B (b), ADP (c) and ADPr (d). Second row:
 reconstructed phase using SHARP (e), SHARP-B (f), ADP (g) and ADPr (h).}
\label{fig4-zoom}
\end{figure*}

\begin{figure}
\begin{center}
\subfigure[SHARP]{\includegraphics[width=.2\textwidth,trim={.8cm .6cm .6cm .8cm},clip]       {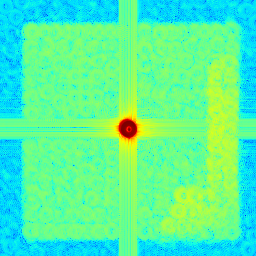}}
\subfigure[SHARP-B]{\includegraphics[width=.2\textwidth,trim={.8cm .6cm .6cm .8cm},clip]{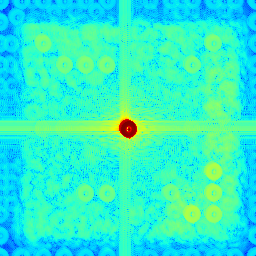}}
\subfigure[ADP]{\includegraphics[width=.2\textwidth,trim={.8cm .6cm .6cm .8cm},clip]{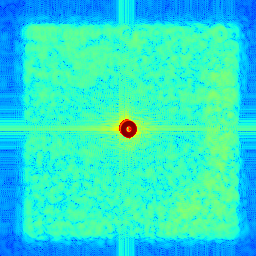}}
\subfigure[ADPr ]{\includegraphics[width=.2\textwidth,trim={.8cm .6cm .6cm .8cm},clip]{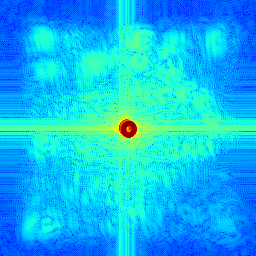}}\\ 
\hskip 1cm
\subfigure{\includegraphics[width=.06\textwidth]{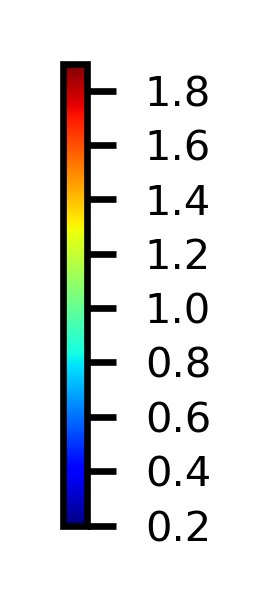}}
\addtocounter{subfigure}{-1}
\hskip .8cm
              \subfigure[SHARP-B]{\includegraphics[width=.2\textwidth,trim={.8cm .6cm .6cm .8cm},clip]{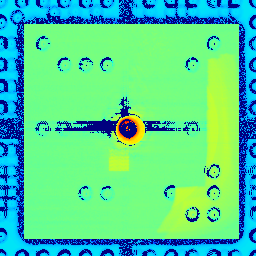}}
            \subfigure[ADP]{\includegraphics[width=.2\textwidth,trim={.8cm .6cm .6cm .8cm},clip]{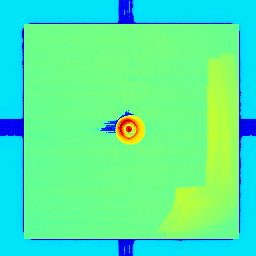}}
            \subfigure[ADPr] {\includegraphics[width=.2\textwidth,trim={.8cm .6cm .6cm .8cm},clip]{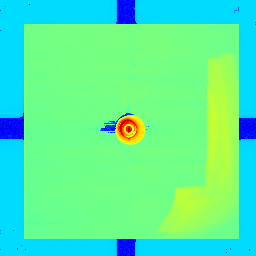}}
\end{center}
\caption{Recovered intensities and background from the experiment presented in Fig.~\ref{fig4}. First row: Recovered intensities by SHARP (a) , SHARP-B (b), ADP (c), and ADPr (d). Second row: recovered backgrounds by SHARP-B (e), ADP (f), and ADPr (g).  {The figures are shown in scale $(\cdot)^{0.05}$ with the colorbar shown on the down left corner}.}
\label{fig4-1}
\end{figure}


\begin{figure}
\label{fig6}
\begin{center}
\subfigure[R-factor]{\includegraphics[width=.45\textwidth]{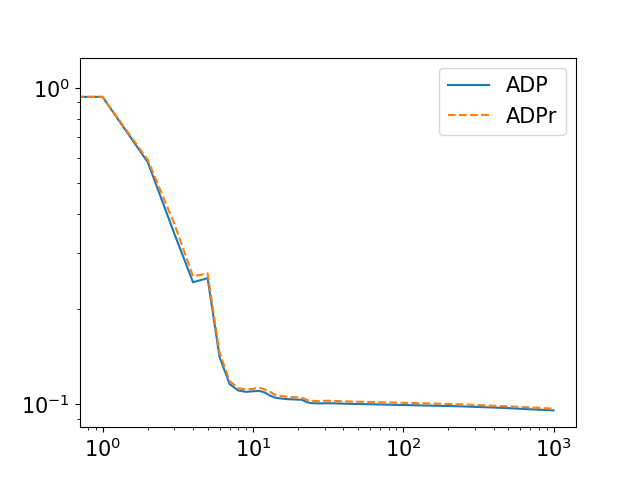}}
\end{center}
\caption{Convergence histories of the R-factor for the proposed algorithm (both ADP and ADPr) when performing the experiment reported in Fig.~\ref{fig4}, using the dataset from~\cite{yu2018three}.}
\label{figRf}   
\end{figure}

\begin{figure}[]
\begin{center}
                                   \subfigure[Amp.] {\includegraphics[width=.45\textwidth,height=.25\textwidth]           {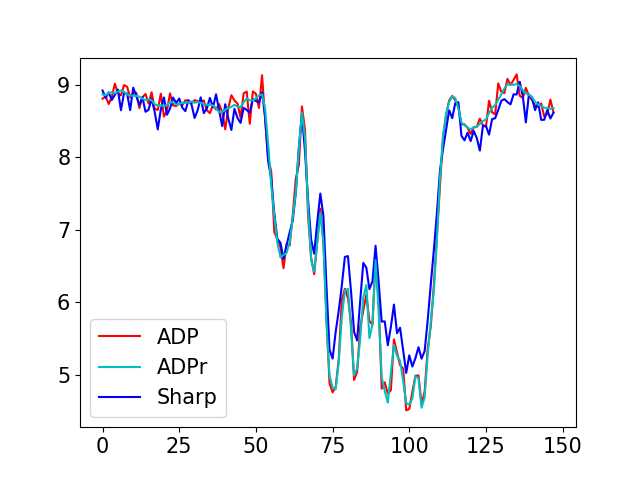}}
                       \subfigure[Phase] {\includegraphics[width=.45\textwidth,height=.25\textwidth]{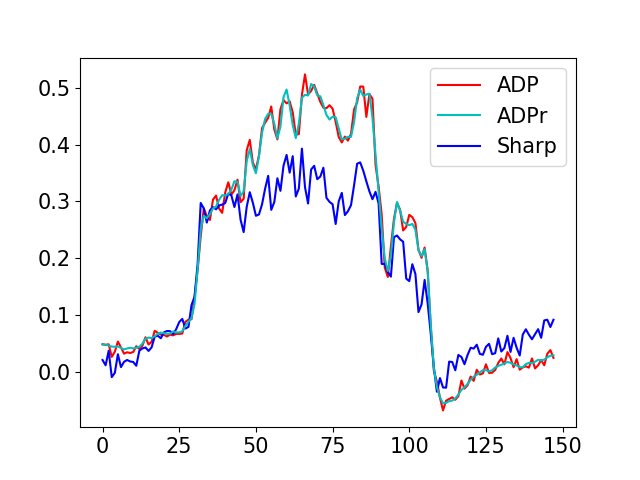}}
\end{center}                     
\caption{Cutlines of the retrieved amplitude (amp.) (a) and phase part (b) from the center of the reconstruction reported in Fig.~\ref{fig4}, using the dataset from~\cite{yu2018three}.}
\label{figCutline} 
\end{figure}


\paragraph{Hard X-ray dataset from PETRA III at DESY}
We also report experimental results using a dataset measured at beamline P06 at PETRA III,  DESY \cite{reinhardt2017beamstop} at a photon energy of  11919 eV. This dataset consists of 2 ptychographic measurements, one without and one with beamstop.  The following results are generated using (1) the data without beamstop (noBS), (2) the data with beamstop (BS), and (3) the merged data from BS and noBS (the low frequencies from noBS plus the high frequency from BS). The datasets are reconstructed using ePIE (extended Ptychographic Iterative Engine \cite{maiden2009improved}), which corresponds to the results presented in~\cite{reinhardt2017beamstop}, and with the proposed ADP algorithm.
For  a fair comparison, all results are produced without position retrieval.  
The amount of noise of this dataset is significantly higher than the one reported in the previous experiments: in this case, the phase contrast produced by the hard X-ray beam  is much weaker than in the soft X-ray dataset. 
Because of this, more iterations are required to achieve reasonable reconstructions. The following results use at most 2000 iterations (we use early-stop for ePIE since it will blow up finally) for both algorithms.

The reconstruction results are shown in Fig. \ref{fig5-1}. The reconstructed phase parts using ePIE is significantly noisy, with lots of blurred out features (Fig. \ref{fig5-1} (a) and (b)). Visually, ADP generates a much cleaner reconstructed phase parts with much better defined features (especially noticeable in the features around the color circles). ADPr in this case produces similar quality results as ADP with no regularization.

We also report an additional reconstruction result using only the BS dataset with ADP (Fig.~\ref{fig5-1} (e)). It is important to note that for this dataset low-frequency information is almost completely lost. We can see how very sharp features are well recovered in the reconstructed phase, while producing a very clean background. This illustrates the robustness of the proposed algorithm even when the measured data is incomplete or contaminated by heavy noise. 
Some features from the BS experiment with ADP seem to be  enhanced, compared with the merged data (specially noticeable again in the color circle areas). We remark that to make our algorithm work well, a good initialization by the illumination from the non-beamstop dataset is adopted. As a future work, we will explore more efficient strategy for the initialization.

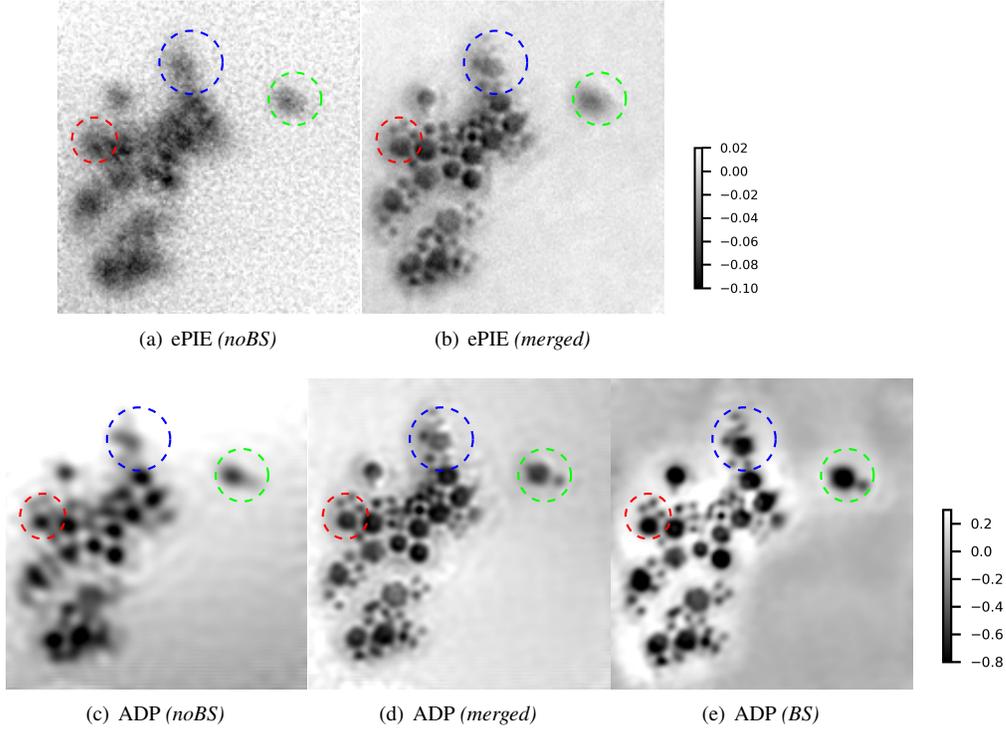
\begin{figure}
\begin{center}
\subfigure[ePIE \textit{(noBS)}]{
     \begin{overpic}[width=.3\textwidth,trim={2.9cm 2.9cm 2.9cm 2.9cm},clip]   {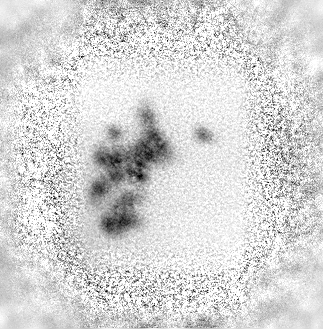}
     \put(2,48){
     \tikz\draw[red,thick,dashed] (0,0) circle (.3);
     }
     \put(65,60){
     \tikz\draw[green,thick,dashed] (0,0) circle (.35);
     }
          \put(30,70){
     \tikz\draw[blue,thick,dashed] (0,0) circle (.42);
     }
     \end{overpic}
}
\hskip -.2cm
\subfigure[ePIE \textit{(merged)}]{
     \begin{overpic}
     [width=.3\textwidth,trim={2.9cm 2.9cm 2.9cm 2.9cm},clip]{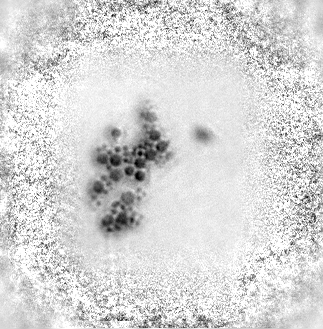}
     \put(2,48){
     \tikz\draw[red,thick,dashed] (0,0) circle (.3);
     }
          \put(65,60){
     \tikz\draw[green,thick,dashed] (0,0) circle (.35);
     }
          \put(30,70){
     \tikz\draw[blue,thick,dashed] (0,0) circle (.42);
     }
     \end{overpic}
}
\subfigure{\includegraphics[width=.1\textwidth]{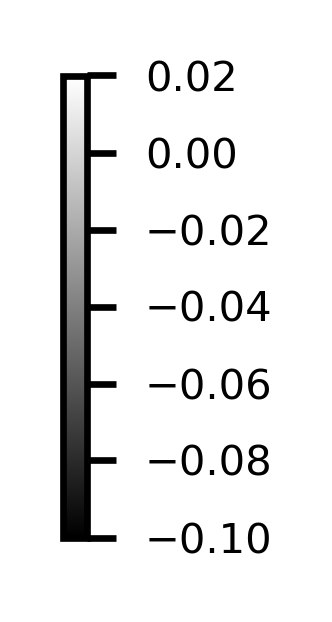}}
\addtocounter{subfigure}{-1}
~\hskip .25\textwidth~
\\
\hskip -1cm
\subfigure[ADP \textit{(noBS)}]{
     \begin{overpic}
     [width=.3\textwidth,trim={2.9cm 2.9cm 2.9cm 2.9cm},clip]   {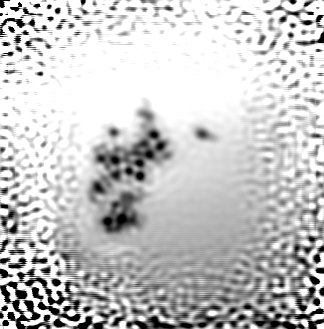}
     \put(2,48){
     \tikz\draw[red,thick,dashed] (0,0) circle (.3);
     }
     \put(65,60){
     \tikz\draw[green,thick,dashed] (0,0) circle (.35);
     }
          \put(30,70){
     \tikz\draw[blue,thick,dashed] (0,0) circle (.42);
     }
     \end{overpic}
}
\hskip -.2cm
\subfigure[ADP \textit{(merged)}]{
     \begin{overpic}
     [width=.3\textwidth,trim={2.95cm 2.9cm 2.85cm 2.9cm},clip]   {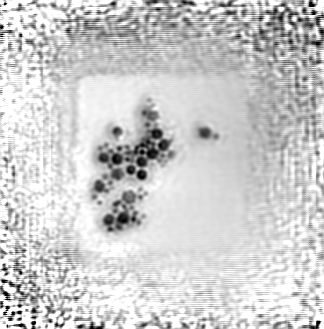}
     \put(2,48){
     \tikz\draw[red,thick,dashed] (0,0) circle (.3);
     }
     \put(65,60){
     \tikz\draw[green,thick,dashed] (0,0) circle (.35);
     }
          \put(30,70){
     \tikz\draw[blue,thick,dashed] (0,0) circle (.42);
     }
     \end{overpic}
}
\hskip -.2cm
\subfigure[ADP \textit{(BS)}]{
     \begin{overpic}
     [width=.3\textwidth,trim={2.8cm 3.05cm 3.0cm 2.75cm},clip]   {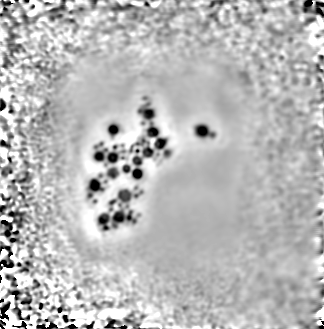}
     \put(2,48){
     \tikz\draw[red,thick,dashed] (0,0) circle (.3);
     }
     \put(65,60){
     \tikz\draw[green,thick,dashed] (0,0) circle (.35);
     }
          \put(30,70){
     \tikz\draw[blue,thick,dashed] (0,0) circle (.42);
     }
     \end{overpic}
}
\subfigure{\includegraphics[width=.1\textwidth]{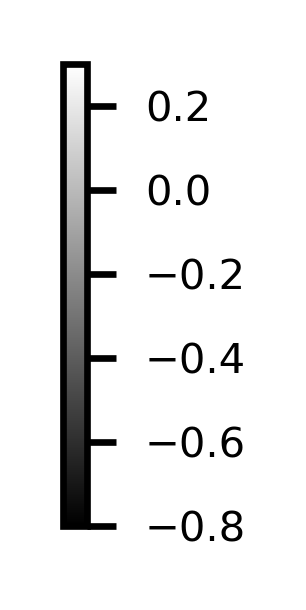}}
\addtocounter{subfigure}{-1}
\end{center}
\caption{Hard X-ray experimental results from the data presented in~\cite{reinhardt2017beamstop}. First row: recovered phase using ePIE, with the noBS dataset (a), and with the merged dataset (b). Second row: recovered phase using ADP, with noBS data (c), merged data (d), and BS data (e).{Results from (a) to (d) are displayed in the range [-0.1,0.02] (colorbar shown in the top right corner) while (e) is depicted in the range [-0.8,0.3] (colorbar shown in the down right corner).}} 
\label{fig5-1}
\end{figure}

{Since our proposed algorithms are implemented by python, which is not specially optimized compared with SHARP (implemented by highly optimized CUDA code). Therefore,  we only give the estimate of computational complexity of proposed ADP as about $63m+2J\times\mathrm{FFT}$ per iteration, while about $78m+2J\times\mathrm{FFT}$ for SHARP-B, that demonstrates the  proposed ADP needs a bit less computational cost than SHARP-B, where $\mathrm{FFT}$ means the computational complexity of 2D fast Fourier Transform for the matrix with size of $\sqrt{\bar m}\times \sqrt{\bar m}$ (same size of the illumination). }

\section{Conclusion}

This paper proposes a novel algorithm for (blind) ptychography reconstruction with integral experimental denoising. All main experimental sources of noise are addressed and characterized in the proposed solution as a mixture of (1) structured parasitic noise, (2) a mix of random noise sources (both Gaussian or Poison), and (3) data outliers. The algorithm is formulated to be robust and efficient, requiring few arithmetic operations and being inherently parallel in most of the steps. The convergence of the method is also proved under mild conditions. Experimental results analyze a variety of datasets with different experimental noise sources and demonstrate how the proposed algorithm achieves superior reconstruction results and denoising than state-of-the-art solutions. As a future work, we will consider partial coherence~\cite{chang2018partially} and position retrieval \cite{marchesini2013augmented} in the model, and also explore additional sparsity techniques~\cite{chang2016general}, deep-learning~\cite{cherukara2018real}, and dictionary-learning methods \cite{chang2018denoising} to further enhance the reconstruction results. 

\section*{Funding}
National Natural Science Foundation of China (Nos.11871372, 11501413, 11271049, 11671002); Natural Science Foundation of Tianjin (No.18JCYBJC16600); 2017-Outstanding Young Innovation Team Cultivation Program (No.043-135202TD1703); Innovation Project (No.043-135202XC1605) of Tianjin Normal University; Tianjin Young Backbone of Innovative Personnel Training Program and Program for Innovative Research Team in Universities of Tianjin (No.TD13-5078); the Center for Applied Mathematics for Energy Research Applications, a joint ASCR-BES funded project within the Office of Science, US Department of Energy, and used resources of the Advanced Light Source, which is a DOE Office of Science User Facility under contract no. DE-AC02-05CH11231; STROBE: A National Science Foundation Science \& Technology Center under Grant No. DMR 1548924.
Parts of this research were carried out at beamline P06 of the light source PETRA III at DESY, a member of the Helmholtz Association (HGF); German Ministry of Education and Research (BMBF) under grant numbers 05K13VK2/05K13OD4 and 05K10VK1/ 05K10OD1, and by VH-VI-403 of the Impuls-und Vernetzungsfonds (IVF) of the Helmholtz Association of German Research Centres;  The Direct Grant of The Chinese University of Hong Kong.




\appendix
\paragraph{\Large Appendix}

\section{Revisit  of SHARP background retrieval algorithm}\label{apdx-1}

The following description was proposed in \cite{marchesini2013augmented}. Given an intensity $I_j$ and for the $k^{th}$ iteration, the background $\phi$ is retrieved as follows:

\begin{equation}
\label{lsBG}
\begin{split}
&\min_{\phi,\eta\geq \eta_{-}} \mathcal H(\eta,\phi):=\tfrac12\textstyle\sum_j\|\eta \circ(I_j-\phi)-|z_j^k|^2\|^2,\\
\end{split}
\end{equation}
given $z^k,$ where a weight function $1\geq \eta\in\mathbb R^{\bar m}_+$ crossing all frames is introduced to be optimized.
By alternating minimization:
\begin{equation}
\eta^{k}=\arg\min_{\eta\geq \eta_-}\textstyle\sum_j\|\eta\circ(I_j-\phi^{k-1})-|z_j^k|^2\|^2=\max\left\{\eta_-, \tfrac{\textstyle\sum_j \langle |z_j^k|^2, I_j-\phi^{k-1}\rangle }{\textstyle\sum_j |I_j-\phi^{k-1}|^2}\right\}.
\end{equation}

Then, by the preconditioned  gradient descent for the background $\phi$:
\begin{equation}
\label{eqeta}
\begin{split}
\phi^{k}&=\phi^{k-1}-\tfrac{1}{J|\eta^k|^2}\nabla_\phi \mathcal H(\eta^k,\phi)\\
&=\tfrac{1}{J}\textstyle\sum_j(I_j-{|z^k_j|^2})+(1-\tfrac{1}{\eta^k})\circ\big(\tfrac{1}{J}\textstyle\sum_j |z_j^k|^2\big),
\end{split}
\end{equation}
where the first term  is essentially the closed form solution for the least squares problem Eq. \eqref{lsBG} with respect to  $\phi$.
When the algorithm converges, $\eta^k\rightarrow 1$ if $k\rightarrow +\infty,$ and Eq. \eqref{eqeta} is a gradient descent scheme with additional stabilization term
$(1-\tfrac{1}{\eta^k})\circ\big(\tfrac{1}{J}\textstyle\sum_j |z_j^k|^2\big),$ and tends to zeros as $k\rightarrow 0$, which  helps to speed up the evolution of background retrieval, heuristically.

\section{Derivations for ADP (Algorithm 1) and ADPr (Algorithm 2)}\label{apdx-0}

\paragraph{ADP}
 The computation of each subproblem is presented here. For simplicity, we omit the superscripts.
For the $\omega$ and $u-$subproblems, with additional proximal terms, we have:
\begin{equation}
\begin{split}
\omega^{opt}&:=\arg\min_{\omega}\mathcal L(\omega, u,z,\mu,\tilde\phi,\Lambda_1,\Lambda_2)+\tfrac{\alpha_1}{2}\|\omega-\hat\omega\|^2\\
 &=\arg\min_\omega  \tfrac12\textstyle\sum\nolimits_j\| \omega\circ\mathcal S_j u-\mathcal F^*(z_j+\Lambda_{1,j}) \|^2+\tfrac{\alpha_1}{2}\|\omega-\hat\omega\|^2,
\end{split}
\end{equation}
and
\begin{equation}
\begin{split}
u^{opt}&:=\arg\min_{u}\mathcal L(\omega, u,z,\mu,\tilde\phi;\Lambda_1,\Lambda_2)+\tfrac{\alpha_2}{2}\|u-\hat u\|^2\\
&=\arg\min_u  \tfrac12\textstyle\sum\nolimits_j\| \omega\circ\mathcal S_j u-\mathcal F^*(z_j+\Lambda_{1,j}) \|^2+\tfrac{\alpha_2}{2}\|u-\hat u\|^2,
\end{split}
\end{equation}
where $\hat\omega$ and $\hat u$ are the approximate solutions in the previous iterations.  By solving a least squares problem, we have:
\begin{equation}
\left\{
\begin{aligned}
&\omega^{opt}=\tfrac{\textstyle\sum_j\mathcal S_j u^*\circ \mathcal F^*(z_j+\Lambda_{1,j})+\alpha_1 \hat \omega}{\textstyle\sum_j |\mathcal S_j u|^2+\alpha_1\bm 1_{\bar m}}
\\
&u^{opt}=\tfrac{\textstyle\sum_j\mathcal S_j^T(\omega^*\circ\mathcal F^*(z_j+\Lambda_{1,j}))+\alpha_2 \hat u}{\textstyle\sum_j\mathcal S_j^T|\omega|^2+\alpha_2\bm 1_{n}}.
\end{aligned}
\right.
\end{equation}

The variables $z_j$ and $\mu_j$ can be determined jointly by solving:
\begin{equation}\label{eqSubOptZ}
\begin{split}
&\min_{z_j,\mu_j} \mathcal G_{\varepsilon,j}(z_j,\mu_j)+\tfrac{r}{2}\|z_j-(\mathcal A_j(\omega,u)-\Lambda_{1,j})\|^2+\tfrac{r}{2}\|\mu_j-(\tilde \phi-\Lambda_{2,j})\|^2,
\end{split}
\end{equation}
with 
\begin{equation}\mathcal G_{\varepsilon,j}(z_j,\mu_j):=\tfrac{1}{2} \big\|\sqrt{C_j}\circ(\sqrt{|z_j|^2+\mu_j^2+\varepsilon^2\bm 1_{\bar m}}- \sqrt{I_j+\varepsilon^2\bm 1_{\bar m}} )\big\|^2
.\end{equation}
By denoting $X_j=(z_j^T,\mu_j^T)^T\in\mathbb C^{2\bar m},$
we have:
\begin{equation}\label{eqSubZ}
X_j^{opt}:=\arg\min\mathcal G_{\varepsilon,j}(X_j)+\tfrac{r}{2}\|X_j-X^0_j\|^2,
\end{equation}
with $X_j^0:=( \mathcal A_j^T(\omega,u)-\Lambda^T_{1,j}, \tilde \phi^T-\Lambda_{2,j}^T)^T.
$
The solution to the above problem has the following forms:
\begin{equation}
X_j^{opt}=((\rho_j^{opt})^T,(\rho_j^{opt})^T)^T\circ \mathrm{sign}(X_j^0),
\end{equation}
with \begin{equation}
\mathrm{sign}(X_j^0):=\left(\tfrac{\mathcal A_j^T(\omega,u)-\Lambda^T_{1,j}} {|X_j^0|^T_*},
\tfrac{\tilde \phi^T-\Lambda^T_{2,j}} {|X_j^0|^T_*}\right)^T,\end{equation} \begin{equation}|X^0_j|_*:= \sqrt{|\mathcal A_j(\omega,u)-\Lambda_{1,j}|^2+ |\tilde \phi-\Lambda_{2,j}|^2},\end{equation}  where
$\rho_j^{opt}$ is determined by:
\begin{equation}
\label{eqRho}
\begin{split}
&\rho_j^{opt}=\arg\min_{\rho_j\in\mathbb R_+^{\bar m}} \mathcal  H_\varepsilon(\rho_j):=
\tfrac12 \langle C_j, \rho_j^2-(I_j+\varepsilon^2 \bm 1_{\bar m})\circ\log(\rho_j^2+\varepsilon^2 \bm 1_{\bar m})  \rangle
 +\tfrac{r}{2}\big\|\rho_j-|X^0_j|_*\big\|^2.
\end{split}
\end{equation}

If $\varepsilon=0,$  the closed form solution \cite{chang2016phase}  to Eq. \eqref{eqRho} is given by:
\begin{equation}
\rho_j^{opt}=\tfrac{r|X^0_j|_*+\sqrt{r^2|X_j^0|_*^2+4(C_j+r\bm 1_{\bar m})\circ C_j\circ I_j}}{2(C_j+r\bm 1_{\bar m})},
\end{equation}
 such that
\begin{equation}\label{eqSolverZ}
X_j^{opt}=\tfrac{r|X^0_j|_*+\sqrt{r^2|X_j^0|_*^2+4(C_j+r\bm 1_{\bar m})\circ C_j\circ I_j}}{2(C_j+r\bm 1_{\bar m})}
\circ \mathrm{sign}(X_j^0).
\end{equation}
Otherwise, we use the projection gradient algorithm to solve it, which gives:
\begin{equation}
\label{eqRho-PG}
\begin{split}
\rho_{j, l+1}
&=\max\{0, \rho_{j,l}-\tau \nabla \mathcal H_\epsilon(\rho_{j,l}) \}\\
&=\max\Big\{0, \Big((1-\tau r)\bm 1_{\bar m}-\tau C_j+\tau C_j\circ\tfrac{{I_j+\varepsilon^2\bm 1_{\bar m}}}{{\rho_{j,l}^2+\varepsilon^2\bm 1_{\bar m}}}\Big)\circ\rho_{j,l}+\tau r |X^0_j|_*\Big\},
\end{split}
\end{equation}
$l=0,1,\cdots,$ with stepsize $\tau$,
since
\begin{equation}
\nabla \mathcal H_\epsilon(\rho_{j})=
\Big(C_j+r\bm 1_{\bar m}-C_j\circ\tfrac{{I_j+\varepsilon^2\bm 1_{\bar m}}}{{\rho_j^2+\varepsilon^2\bm 1_{\bar m}}}\Big)\circ\rho_j-r|X^0_j|_*.
\end{equation}





For the $\tilde\phi-$subproblem, the solution is given directly as
\begin{equation}
\tilde\phi^{opt}=\arg\min_{\tilde\phi}\mathcal L(\omega, u,z,\mu,\tilde\phi,\Lambda_1,\Lambda_2)=\tfrac{1}{J}\sumMy(\mu_{j}+\Lambda_{2,j}).
\end{equation}

We can now further simplify the proposed algorithm. Since
$0=\tilde\phi^{k+1}-\tfrac{1}{J}\textstyle\sum\nolimits_j (\mu^{k+1}_j+\Lambda^k_{2,j}),$ and
$\Lambda^{k+1}_{2,j}=\Lambda^k_{2,j}+\mu^{k+1}_j-\tilde \phi^{k+1}
$, we have
$
\textstyle\sum\nolimits_j  \Lambda^{k+1}_{2,j}=0,
$
such that
$
\tilde\phi^{k+1}=\tfrac{1}{J}\textstyle\sum\nolimits_j \mu^{k+1}_j.
$
That is to say, one can replace the variable $\tilde \phi^k$ by the pointwise average $\tfrac{1}{J}\sumMy_j \mu_j^k$ and remove the subproblem with respect to $\tilde\phi$.

\paragraph{ADPr}\label{apdx0-1}
Similarly, we omit the superscripts for simplicity.
We consider the $u-$subproblem, which is expressed below:
\begin{equation}
\begin{split}
&\min \tfrac{r}{2}\textstyle\sum\nolimits_j\| \omega\circ\mathcal S_j u-\mathcal F^*(z_j+\Lambda_{1,j}) \|^2+\tfrac{\beta}{2}\textstyle\sum\nolimits_j \|\Lambda_{3,j}+p_j-\nabla \mathcal S_j u\|^2+\tfrac{\alpha_2}{2}\|u-\hat u\|^2,
\end{split}
\end{equation}
where $\hat u$ is the previous iterative solution.
The minimizer to the above optimization problem satisfies the following equation:
\begin{equation}\label{eqSubUReg}
\begin{split}
&\qquad\qquad\textstyle\sum\nolimits_j\big( \mathrm{diag}(r\mathcal S_j^T |\omega|^2+\alpha_2 \bm 1_{n})-\beta \mathcal S_j^T \Delta\mathcal S_j\big) u\\
&=\textstyle\sum\nolimits_j\big( r\mathcal S_j^T(\omega^*\circ \mathcal F^*(z_j+\Lambda_{1,j})) -\beta\mathcal S_j^T\mathrm{div}(p_j+\Lambda_{3,j})\big)+\alpha_2 \hat u,
\end{split}
\end{equation}
where $\mathrm{div}$ denotes the divergence operator satisfying $\mathrm{div}=-\nabla^T,$ and $\Delta$ denotes the Laplacian operators satifying $\Delta=\mathrm{div(\nabla)}$. One can readily verify that the coefficient matrix is Hermitian,
 and positive definite, i.e.:
 \begin{equation}
 \begin{split}
 &\langle\big(\mathrm{diag}(r\mathcal S_j^T |\omega|^2+\alpha_2\bm 1_{n})-\beta \mathcal S_j^T \Delta\mathcal S_j\big) v,v\rangle\\
=&
 r\langle \mathcal S_j^T \omega \circ v,\mathcal S_j^T \omega\circ v\rangle  +\alpha_2\|v\|^2+\beta\langle  \nabla\mathcal S_j v,\nabla\mathcal S_j v\rangle \\
 =&r\|\mathcal S_j^T \omega \circ v\|^2  +\alpha_2\|v\|^2+\beta\| \nabla\mathcal S_j v\|^2 > 0,~\forall 0\not=v\in\mathbb C^n,
 \end{split}
 \end{equation}
 such that we adopt conjugate gradient (CG) algorithm to solve Eq. \eqref{eqSubUReg}.

The $p_j-$subproblem has the following closed form solution:
\begin{equation}
p^{opt}_j=\max\{ 0, |\nabla\mathcal S_j u-\Lambda_{3,j}|-\tfrac{\lambda}{\beta}\bm 1_{\bar m} \}\circ \tfrac{\nabla\mathcal S_j u-\Lambda_{3,j}}{|\nabla\mathcal S_j u-\Lambda_{3,j}|}.
\end{equation}

The subproblems with respect to other variables are mainly the same as in the previous subsection, and we omit the details.

\section{Proof of Theorem 1}\label{apdx-2}
The relation of the stationary points is given in the following lemma:
\begin{lem}
Iterative solutions  $\{(\omega^k, u^{k}, z^{k},\mu^{k},\Lambda^{k}_1,\Lambda^{k}_2)\}_{k}$ satisfy the following relations:
\begin{equation}
\label{eqStationary}
\left\{
\begin{aligned}
&0=\omega^{k+1}\circ \textstyle\sum_j|\mathcal S_j u^k|^2-\textstyle\sum_j(\mathcal S_j u^k)^*\circ \mathcal F^*(z_j^k+\Lambda_{1,j}^k)+\alpha_1(\omega^{k+1}-\omega^k);\\
&0=\textstyle\sum\nolimits_j \mathcal S_j^T |\omega^{k+1}|^2\circ u^{k+1}-\textstyle\sum\nolimits_j \mathcal S_j^T ((\omega^{k+1})^*\circ\mathcal F^{*}(z_j^k+\Lambda^k_{1,j}))+\alpha_2(u^{k+1}-u^k);\\
&0=\nabla \mathcal G_\varepsilon(z^{k+1},\mu^{k+1})+r((\Lambda^{k+1}_1)^T,(\Lambda^{k+1}_2)^T)^T;\\
&0=-\Lambda^{k+1}_{1,j}+\Lambda^k_{1,j}+z^{k+1}_j-\mathcal A_j(\omega^{k+1},u^{k+1});\\
&0=-\Lambda^{k+1}_{2,j}+\Lambda^k_{2,j}+\mu^{k+1}_j-\tfrac{1}{J}\textstyle\sum\nolimits_j \mu^{k+1}_j.
\end{aligned}
\right.
\end{equation}
\end{lem}

\begin{proof}
 By  Eq. \eqref{eqSolverZ}, one has:
\begin{equation}\label{eqPos}
\sqrt{|z_j^{k+1}|^2+|\mu_j^{k+1}|^2}=\tfrac{\sqrt{I_j}+r|Y_j^k|}{1+r}\geq \tfrac{\sqrt{I_j}}{1+r}>0,
\end{equation}
such that the objective function in Eq. \eqref{eqSubZ} is smooth at $X_j^{k+1}:=(z_j^{k+1},\mu_j^{k+1})$.
%
Then one can readily derive Eq. \eqref{eqStationary} by calculating the first order derivative for each subproblem.
\end{proof}


\begin{lem}\label{lem0}

\begin{equation}
\|\nabla \mathcal G_\varepsilon(z^1,\mu^1)-\nabla \mathcal G_\varepsilon(z^2,\mu^2)\|\leq 
L_\varepsilon(\|z^1-z^2\|+\|\mu^1-\mu^2\|),
\end{equation}
where $L_\epsilon$ is a positive constant independent with $z^1,\mu^1,z^2,\mu^2.$
\end{lem}
\begin{proof}
It can be readily proved following \cite{chang2018Blind}, and we omit the details here.
\end{proof}

\begin{lem}
\label{lemDes}
Denoting $Y^k:=(\omega^k,u^k,z^k,\mu^k,\tilde \phi^k,\Lambda_1^k,\Lambda_2^k)$,
if $r>4L_\varepsilon$,
\begin{equation}
\begin{split}
\mathcal L(Y^k)-\mathcal L(Y^{k+1})&\geq  c_{\varepsilon, r} \|Y^{k+1}-Y^k\|^2,
\end{split}
\end{equation}
where $c_{\varepsilon, r}$ is a positive parameter independent from $\{Y^k\}$.
\end{lem}
\begin{proof}
Regarding the $\omega$ and $u-$subproblems:
\begin{equation}\label{eqEnergyU}
\begin{split}
&\qquad\mathcal L(Y^k)-\mathcal L(\omega^{k+1},u^{k+1},z^k,\mu^k,\tilde \phi^k,\Lambda_1^k,\Lambda_2^k)\\
&\geq\tfrac{1}{2}\|\mathcal A(\omega^{k+1}-\omega^k, u^k)\|^2+\tfrac{\alpha_1}{2}\|\omega^{k+1}-\omega^k\|^2\\
&+\tfrac{1}{2}\|\mathcal A(\omega^{k+1}, u^{k+1}-u^k)\|^2+\tfrac{\alpha_2}{2}\|u^{k+1}-u^k\|^2.\\
\end{split}
\end{equation}

For the $z$ and $\mu$ subproblems, by Eq. \eqref{eqSubOptZ} and the Lemma \ref{lem0}, we have:
\begin{equation}\label{eqEnergyZM}
\begin{split}
&\mathcal L(\omega^{k+1},u^{k+1},z^k,\mu^k,\tilde \phi^k,\Lambda_1^k,\Lambda_2^k)-\mathcal L(\omega^{k+1},u^{k+1},z^{k+1},\mu^{k+1},\tilde \phi^k,\Lambda_1^k,\Lambda_2^k)\\
&\geq \tfrac{r-3L_\varepsilon}{2}(\|z^{k+1}_j-z_j^k\|^2+\|\mu^{k+1}_j-\mu_j^k\|^2).
\end{split}
\end{equation}

For the subproblems of $\tilde \phi,$ we have:
\begin{equation}\label{eqEnergyP}
\begin{split}
&\mathcal L(\omega^{k+1},u^{k+1},z^{k+1},\mu^{k+1},\tilde \phi^k,\Lambda_1^k,\Lambda_2^k)-\mathcal L(\omega^{k+1},u^{k+1},z^{k+1},\mu^{k+1},\tilde \phi^{k+1},\Lambda_1^k,\Lambda_2^k)\\
&\geq \tfrac{r J}{2}\|\tilde\phi^{k+1}-\tilde\phi^k\|^2.
\end{split}
\end{equation}

Further we have:
\begin{equation}
\label{eqEnergyMul}
\begin{split}
&\mathcal L(\omega^{k+1},u^{k+1},z^{k+1},\mu^{k+1},\tilde \phi^{k+1},\Lambda_1^k,\Lambda_2^k)-\mathcal L(Y^{k+1})\\
&\stackrel{\text{Eq.} \eqref{eqUpMul}}{=}-r\textstyle\sum_j (\|\Lambda_{1,j}^{k+1}-\Lambda_{1,j}^{k}\|^2+\|\Lambda_{2,j}^{k+1}-\Lambda_{2,j}^{k}\|^2)\\
&\stackrel{\text{Eq.} \eqref{eqStationary}}{\geq} -\tfrac{2L^2_\varepsilon}{r}(\|z^{k+1}-z^k\|^2+\|\omega^{k+1}-\omega^k\|^2),
\end{split}
\end{equation}
where the last relation is derived by Lemma \ref{lem0}.

Finally, by summing up Eqs. \eqref{eqEnergyU}-\eqref{eqEnergyMul} together, we get:
\begin{equation}
\begin{split}
\mathcal L(Y^k)-\mathcal L(Y^{k+1})\geq \min\left\{\tfrac{\alpha_1}{2},\tfrac{\alpha_2}{2},\tfrac{r-3L_\varepsilon}{2}-\tfrac{2L^2_\varepsilon}{r},\tfrac{rJ}{2} \right\}\|Y^{k+1}-Y^k\|^2.
\end{split}
\end{equation}
Setting $r>4L_\varepsilon$ concludes this lemma.
\end{proof}
\vskip .2in


\begin{proof}[Proof of Theorem \ref{thm1}]

We first show that $\{\mathcal L(Y^k)\}$ is lower bounded:
\begin{equation}
\begin{split}
&\mathcal L(Y^{k+1})=\mathcal G_\varepsilon(z^{k+1},\mu^{k+1})-\Re\langle \nabla_z \mathcal G_{\varepsilon,j} (z_j^{k+1},\mu_j^{k+1}),z^{k+1}_j-\mathcal A_j(\omega^{k+1},u^{k+1})\rangle\\
&\qquad\qquad-\langle \nabla_\mu \mathcal G_{\varepsilon,j}(z_j^{k+1},\mu_j^{k+1}) ,\mu^{k+1}_j-\tilde \phi^{k+1}\rangle+\tfrac{r}{2}\|z^{k+1}_j-\mathcal A_j(\omega^{k+1},u^{k+1})\|^2\\
&\qquad\qquad+\tfrac{r}{2}\|\mu^{k+1}_j-\tilde \phi^{k+1}\|^2\\
&=\mathcal G_\varepsilon(z^{k+1},\mu^{k+1})-\mathcal G_\varepsilon(\mathcal A_j(\omega^{k+1},u^{k+1}),\tilde\phi^{k+1})\\
&-\Re\langle \nabla \mathcal G_{\varepsilon,j} (z_j^{k+1},\mu_j^{k+1}),( (z^{k+1}_j-\mathcal A_j(\omega^{k+1},u^{k+1}))^T, (\mu^{k+1}_j-\tilde \phi^{k+1})^T)^T\rangle\\
&+\tfrac{r}{2}\|z^{k+1}_j-\mathcal A_j(\omega^{k+1},u^{k+1})\|^2+\tfrac{r}{2}\|\mu^{k+1}_j-\tilde \phi^{k+1}\|^2+\mathcal G_\varepsilon(\mathcal A_j(\omega^{k+1},u^{k+1}),\tilde\phi^{k+1})\\
&\geq \tfrac{r-L_\varepsilon}{2}\big(\|z^{k+1}_j-\mathcal A_j(\omega^{k+1},u^{k+1})\|^2+\|\mu^{k+1}_j-\tilde \phi^{k+1}\|^2\big)+\mathcal G_\varepsilon(\mathcal A_j(\omega^{k+1},u^{k+1}),\tilde\phi^{k+1}),
\end{split}
\end{equation}
where the last relation is derived from the Lemma \ref{lem0}.
Therefore, $\mathcal L(Y^k)>-\infty$ with $r\geq 4 L_\varepsilon.$
Furthermore, by the Lemma \ref{lemDes}:
\begin{equation}
\begin{split}
&+\infty> \mathcal L(Y^0)-\mathcal L(Y^k)=\sum_{l=0}^k(\mathcal L(Y^l)-\mathcal L(Y^{l+1}))\geq  c_{\varepsilon, r} \sum_{l=0}^k\|Y^{l+1}-Y^l\|^2.\\
\end{split}
\end{equation}
Accordingly we get:
\begin{equation}
\label{eqSucse}
\begin{split}
&\lim_{l\rightarrow +\infty} \|Y^{l+1}-Y^l\|=0.\\
\end{split}
\end{equation}

For any limit point  $Y^\star:=(\omega^\star, u^\star, z^\star,\mu^\star,\tilde \phi^\star,\Lambda^\star_1,\Lambda^\star_2)$ of the iterative sequence $\{Y^k\}$, there exists a subsequence $Y^{n_k}:=(\omega^{n_k}, u^{n_k}, z^{n_k},\mu^{n_k},\tilde \phi^{n_k},\Lambda^{n_k}_1,\Lambda^{n_k}_2)\subset Y^k,$
\text{such that}
\begin{equation}
\lim_{k\rightarrow+\infty} Y^{n_k}=Y^\star.
\end{equation}
We can see that $\{Y^{n_k}\}$ is bounded such that one can derive:
$\lim_{k\rightarrow+\infty}\mathcal A(\omega^{n_k},u^{n_k})=\mathcal A(\omega^\star,u^\star),$
and
$\lim_{k\rightarrow+\infty}\textstyle\sum\nolimits_j \mathcal S_j^T |\omega^{n_k}|^2\circ u^{n_k}=\textstyle\sum\nolimits_j \mathcal S_j^T |\omega^{\star}|^2\circ u^{\star}.
$ 
By Eq. \eqref{eqSucse}, $\{Y^{n_k-1}\}$ is bounded as well.  Hence
\begin{equation}
\lim_{k\rightarrow +\infty}\omega^{n_k}\circ \textstyle\sum_j|\mathcal S_j u^{n_k-1}|^2=\omega^{\star}\circ \textstyle\sum_j|\mathcal S_j u^{\star}|^2,~~
\end{equation}

\begin{equation}
\lim_{k\rightarrow+\infty}\textstyle\sum_j(\mathcal S_j u^{n_k-1})^*\circ \mathcal F^*(z_j^{n_k-1}+\Lambda_{1,j}^{n_k-1})=\textstyle\sum_j(\mathcal S_j u^{\star})^*\circ \mathcal F^*(z_j^{\star}+\Lambda_{1,j}^{\star}),
\end{equation}
and
\begin{equation}
\lim_{k\rightarrow+\infty}\textstyle\sum\nolimits_j \mathcal S_j^T ((\omega^{n_k})^*\circ\mathcal F^{*}(z_j^{n_k-1}+\Lambda^{n_k-1}_{1,j}))=\textstyle\sum\nolimits_j \mathcal S_j^T ((\omega^{\star})^*\circ\mathcal F^{*}(z_j^{\star}+\Lambda^{\star}_{1,j})).
\end{equation}

By Lemma \ref{lem0}:
\begin{equation}
\lim_{k\rightarrow+\infty}\nabla\mathcal G_\varepsilon(z^{n_k},\mu^{n_k})=\nabla\mathcal G_\varepsilon(z^{\star},\mu^{\star}).
\end{equation}

Finally, since
\begin{equation}
\label{eqStationary-subset}
\begin{split}
&0=\omega^{n_k}\circ \textstyle\sum_j|\mathcal S_j u^{n_k-1}|^2-\textstyle\sum_j(\mathcal S_j u^{n_k-1})^*\circ \mathcal F^*(z_j^{n_k-1}+\Lambda_{1,j}^{n_k-1})+\alpha_1(\omega^{n_k}-\omega^{n_k-1});\\
&0=\textstyle\sum\nolimits_j \mathcal S_j^T |\omega^{n_k}|^2\circ u^{n_k}-\textstyle\sum\nolimits_j \mathcal S_j^T ((\omega^{n_k})^*\circ\mathcal F^{*}(z_j^{n_k-1}+\Lambda^{n_k-1}_{1,j}))+\alpha_2(u^{n_k}-u^{n_k-1});\\
&0=\nabla \mathcal G_\varepsilon(z^{n_k},\mu^{n_k})+r((\Lambda^{n_k}_1)^T,(\Lambda^{n_k}_2)^T)^T;\\
&0=-\Lambda^{n_k}_{1,j}+\Lambda^{n_k-1}_{1,j}+z^{n_k}_j-\mathcal A_j(\omega^{n_k},u^{n_k});\\
&0=-\Lambda^{n_k}_{2,j}+\Lambda^{n_k-1}_{2,j}+\mu^{n_k}_j-\tfrac{1}{J}\textstyle\sum\nolimits_j \mu^{n_k}_j,
\end{split}
\end{equation}
following the above calculations of these limits, we can get:
\begin{equation}
\left\{
\begin{aligned}
&0=\omega^\star\circ \textstyle\sum_j|\mathcal S_j u^\star|^2-\textstyle\sum_j(\mathcal S_j u^\star)^*\circ \mathcal F^*(z_j^\star+\Lambda_{1,j}^*);\\
&0=\textstyle\sum\nolimits_j \mathcal S_j^T |\omega^\star|^2\circ u^{\star}-\textstyle\sum\nolimits_j \mathcal S_j^T((\omega^\star)^*\circ\mathcal F^{*}(z_j^{\star}+\Lambda^{\star}_{1,j}));\\
&0=\nabla \mathcal G_\varepsilon(z^{\star},\mu^{\star})+r((\Lambda^{\star}_1)^T,(\Lambda^{\star}_2)^T)^T;\\
&0=z^{\star}_j-\mathcal A_j(\omega^\star,u^{\star});\\
&0=\mu^{\star}_j-\tfrac{1}{J}\textstyle\sum\nolimits_j\mu^{\star}_j,
\end{aligned}
\right.
\end{equation}
which immediately implies that  $Y^\star$ is the stationary point of Eq. \eqref{eqGenModel-II} and concludes this theorem.
\end{proof}

\bibliography{rD}

\begin{thebibliography}{10}
\newcommand{\enquote}[1]{``#1''}

\bibitem{nellist1995resolution}
P.~Nellist, B.~McCallum, and J.~Rodenburg, \enquote{Resolution beyond the
  `information limit' in transmission electron microscopy,} Nature
  \textbf{374}, 630 (1995).

\bibitem{rodenburg2004phase}
J.~M. Rodenburg and H.~M. Faulkner, \enquote{A phase retrieval algorithm for
  shifting illumination,} Applied physics letters \textbf{85}, 4795--4797
  (2004).

\bibitem{maiden2009improved}
A.~M. Maiden and J.~M. Rodenburg, \enquote{An improved ptychographical phase
  retrieval algorithm for diffractive imaging,} Ultramicroscopy \textbf{109},
  1256--1262 (2009).

\bibitem{jiang2018electron}
Y.~Jiang, Z.~Chen, Y.~Han, P.~Deb, H.~Gao, S.~Xie, P.~Purohit, M.~W. Tate,
  J.~Park, S.~M. Gruner, V.~Elser, and D.~A. Muller, \enquote{Electron
  ptychography of 2d materials to deep sub-{{\AA}}ngstr{\"o}m resolution,}
  Nature \textbf{559}, 343 (2018).

\bibitem{shi2016soft}
X.~Shi, P.~Fischer, V.~Neu, D.~Elefant, J.~Lee, D.~Shapiro, M.~Farmand,
  T.~Tyliszczak, H.-W. Shiu, S.~Marchesini, and S.~Roy, \enquote{Soft x-ray
  ptychography studies of nanoscale magnetic and structural correlations in
  thin {SmCo}$_5$ films,} Applied Physics Letters \textbf{108}, 094103 (2016).

\bibitem{giewekemeyer2010quantitative}
K.~Giewekemeyer, P.~Thibault, S.~Kalbfleisch, A.~Beerlink, C.~M. Kewish,
  M.~Dierolf, F.~Pfeiffer, and T.~Salditt, \enquote{Quantitative biological
  imaging by ptychographic x-ray diffraction microscopy,} Proceedings of the
  National Academy of Sciences \textbf{107}, 529--534 (2010).

\bibitem{shapiro2014chemical}
D.~A. Shapiro, Y.-S. Yu, T.~Tyliszczak, J.~Cabana, R.~Celestre, W.~Chao,
  K.~Kaznatcheev, A.~D. Kilcoyne, F.~Maia, S.~Marchesini, Y.~S. Meng,
  T.~Warwick, L.~L. Yang, and H.~Padmore, \enquote{Chemical composition mapping
  with nanometre resolution by soft x-ray microscopy,} Nature Photonics
  \textbf{8}, 765--769 (2014).

\bibitem{yu2018three}
Y.-S. Yu, M.~Farmand, C.~Kim, Y.~Liu, C.~P. Grey, F.~C. Strobridge,
  T.~Tyliszczak, R.~Celestre, P.~Denes, J.~Joseph, H.~Krishnan, F.~R. N.~C.
  Maia, A.~L.~D. Kilcoyne, S.~Marchesini, T.~P.~C. Leite, T.~Warwick,
  H.~Padmore, J.~Cabana, and D.~A. Shapiro, \enquote{Three-dimensional
  localization of nanoscale battery reactions using soft x-ray tomography,}
  Nature communications \textbf{9}, 921 (2018).

\bibitem{holler2017high}
M.~Holler, M.~Guizar-Sicairos, E.~H. Tsai, R.~Dinapoli, E.~M{\"u}ller, O.~Bunk,
  J.~Raabe, and G.~Aeppli, \enquote{High-resolution non-destructive
  three-dimensional imaging of integrated circuits,} Nature \textbf{543},
  402--406 (2017).

\bibitem{wiedorn2017post}
M.~O. Wiedorn, S.~Awel, A.~J. Morgan, M.~Barthelmess, R.~Bean, K.~R. Beyerlein,
  L.~M.~G. Chavas, N.~Eckerskorn, H.~Fleckenstein, M.~Heymann, D.~A. Horke,
  J.~Kno{\v{s}}ka, V.~Mariani, D.~Oberth{\"{u}}r, N.~Roth, O.~Yefanov,
  A.~Barty, S.~Bajt, J.~K{\"{u}}pper, A.~V. Rode, R.~A. Kirian, and H.~N.
  Chapman, \enquote{Post-sample aperture for low background diffraction
  experiments at x-ray free-electron lasers,} Journal of Synchrotron Radiation
  \textbf{24}, 1296--1298 (2017).

\bibitem{reinhardt2017beamstop}
J.~Reinhardt, R.~Hoppe, G.~Hofmann, C.~D. Damsgaard, J.~Patommel, C.~Baumbach,
  S.~Baier, A.~Rochet, J.-D. Grunwaldt, G.~Falkenberg, and C.~G. Schroer,
  \enquote{Beamstop-based low-background ptychography to image weakly
  scattering objects,} Ultramicroscopy \textbf{173}, 52--57 (2017).

\bibitem{wang2017background}
C.~Wang, Z.~Xu, H.~Liu, Y.~Wang, J.~Wang, and R.~Tai, \enquote{Background noise
  removal in x-ray ptychography,} Applied optics \textbf{56}, 2099--2111
  (2017).

\bibitem{thibault2009probe}
P.~Thibault, M.~Dierolf, O.~Bunk, A.~Menzel, and F.~Pfeiffer, \enquote{Probe
  retrieval in ptychographic coherent diffractive imaging,} Ultramicroscopy
  \textbf{109}, 338--343 (2009).

\bibitem{hesse2015proximal}
R.~Hesse, D.~R. Luke, S.~Sabach, and M.~K. Tam, \enquote{Proximal heterogeneous
  block implicit-explicit method and application to blind ptychographic
  diffraction imaging,} SIAM Journal on Imaging Sciences \textbf{8}, 426--457
  (2015).

\bibitem{chang2018Blind}
H.~Chang, P.~Enfedaque, and S.~Marchesini, \enquote{Blind ptychographic phase
  retrieval via convergent alternating direction method of multipliers,} SIAM
  Journal on Imaging Sciences \textbf{12}, 153--185 (2019).

\bibitem{Elser2003}
V.~Elser, \enquote{Phase retrieval by iterated projections,} J. Opt. Soc. Am. A
  \textbf{20}, 40--55 (2003).

\bibitem{thibault2012maximum}
P.~Thibault and M.~Guizar-Sicairos, \enquote{Maximum-likelihood refinement for
  coherent diffractive imaging,} New Journal of Physics \textbf{14}, 063004
  (2012).

\bibitem{Luke2005}
D.~R. Luke, \enquote{Relaxed averaged alternating reflections for diffraction
  imaging,} Inverse Probl. \textbf{21}, 37--50 (2005).

\bibitem{marchesini2016sharp}
S.~Marchesini, H.~Krishnan, D.~A. Shapiro, T.~Perciano, J.~A. Sethian, B.~J.
  Daurer, and F.~R. Maia, \enquote{{SHARP}: a distributed, {GPU}-based
  ptychographic solver,} Journal of Applied Crystallography \textbf{49},
  1245--1252 (2016).

\bibitem{glowinski1989augmented}
R.~Glowinski and P.~Le~Tallec, \emph{Augmented Lagrangian and
  operator-splitting methods in nonlinear mechanics} (Philadelphia, PA: SIAM,
  1989).

\bibitem{Wu&Tai2010}
C.~Wu and X.-C. Tai, \enquote{Augmented {Lagrangian} method, dual methods and
  split-{Bregman} iterations for {ROF}, vectorial {TV} and higher order
  models,} SIAM J. Imaging Sci. \textbf{3}, 300--339 (2010).

\bibitem{boyd2011distributed}
S.~Boyd, N.~Parikh, E.~Chu, B.~Peleato, and J.~Eckstein, \enquote{Distributed
  optimization and statistical learning via the alternating direction method of
  multipliers,} Foundations and Trends in Machine learning \textbf{3}, 1--122
  (2011).

\bibitem{godard2012noise}
P.~Godard, M.~Allain, V.~Chamard, and J.~Rodenburg, \enquote{Noise models for
  low counting rate coherent diffraction imaging,} Optics express \textbf{20},
  25914--25934 (2012).

\bibitem{odstrvcil2018iterative}
M.~Odstr{\v{c}}il, A.~Menzel, and M.~Guizar-Sicairos, \enquote{Iterative
  least-squares solver for generalized maximum-likelihood ptychography,} Optics
  express \textbf{26}, 3108--3123 (2018).

\bibitem{enders2014ptychography}
B.~Enders, M.~Dierolf, P.~Cloetens, M.~Stockmar, F.~Pfeiffer, and P.~Thibault,
  \enquote{Ptychography with broad-bandwidth radiation,} Applied Physics
  Letters \textbf{104}, 171104 (2014).

\bibitem{marchesini2013augmented}
S.~Marchesini, A.~Schirotzek, C.~Yang, H.-T. Wu, and F.~Maia,
  \enquote{Augmented projections for ptychographic imaging,} Inverse Problems
  \textbf{29}, 115009 (2013).

\bibitem{daurer2017nanosurveyor}
B.~J. Daurer, H.~Krishnan, T.~Perciano, F.~R. Maia, D.~A. Shapiro, J.~A.
  Sethian, and S.~Marchesini, \enquote{Nanosurveyor: a framework for real-time
  data processing,} Advanced structural and chemical imaging \textbf{3}, 7
  (2017).

\bibitem{chang2016Total}
H.~Chang, Y.~Lou, Y.~Duan, and S.~Marchesini, \enquote{Total variation--based
  phase retrieval for {Poisson} noise removal,} SIAM Journal on Imaging
  Sciences \textbf{11}, 24--55 (2018).

\bibitem{chang2018denoising}
H.~Chang and S.~Marchesini, \enquote{Denoising {Poisson} phaseless measurements
  via orthogonal dictionary learning,} Optics Express \textbf{26}, 19773--19796
  (2018).

\bibitem{murtagh1995image}
F.~Murtagh, J.-L. Starck, and A.~Bijaoui, \enquote{Image restoration with noise
  suppression using a multiresolution support.} Astronomy and Astrophysics
  Supplement Series \textbf{112}, 179 (1995).

\bibitem{chakrabarti2012image}
A.~Chakrabarti and T.~Zickler, \enquote{Image restoration with signal-dependent
  camera noise,} arXiv:1204.2994  (2012).

\bibitem{li2015reweighted}
J.~Li, Z.~Shen, R.~Yin, and X.~Zhang, \enquote{A reweighted $l^2$ method for
  image restoration with poisson and mixed poisson-gaussian noise,} Inverse
  Probl. Imaging \textbf{9}, 875--894 (2015).

\bibitem{kirby2013low}
N.~M. Kirby, S.~T. Mudie, A.~M. Hawley, D.~J. Cookson, H.~D. Mertens,
  N.~Cowieson, and V.~Samardzic-Boban, \enquote{A low-background-intensity
  focusing small-angle x-ray scattering undulator beamline,} Journal of Applied
  Crystallography \textbf{46}, 1670--1680 (2013).

\bibitem{chang2017variational}
H.~Chang, S.~Marchesini, Y.~Lou, and T.~Zeng, \enquote{Variational phase
  retrieval with globally convergent preconditioned proximal algorithm,} SIAM
  Journal on Imaging Sciences \textbf{11}, 56--93 (2018).

\bibitem{wen2012}
Z.~Wen, C.~Yang, X.~Liu, and S.~Marchesini, \enquote{Alternating direction
  methods for classical and ptychographic phase retrieval,} Inverse Probl.
  \textbf{28}, 115010 (2012).

\bibitem{chang2016phase}
H.~Chang, Y.~Lou, M.~K. Ng, and T.~Zeng, \enquote{Phase retrieval from
  incomplete magnitude information via total variation regularization,} SIAM
  Journal on Scientific Computing \textbf{38}, A3672--A3695 (2016).

\bibitem{chang2018partially}
H.~Chang, P.~Enfedaque, Y.~Lou, and S.~Marchesini, \enquote{Partially coherent
  ptychography by gradient decomposition of the probe,} Acta Crystallographica
  Section A: Foundations and Advances \textbf{74}, 157--169 (2018).

\bibitem{chang2016general}
H.~Chang and S.~Marchesini, \enquote{A general framework for denoising
  phaseless diffraction measurements,} arXiv preprint arXiv:1611.01417  (2016).

\bibitem{cherukara2018real}
M.~J. Cherukara, Y.~S. Nashed, and R.~J. Harder, \enquote{Real-time coherent
  diffraction inversion using deep generative networks,} arXiv preprint
  arXiv:1806.03992  (2018).

\end{thebibliography}

\end{document}